\newif\ifdraft
\draftfalse

\newif\ifarxiv
\newif\ifone

\onefalse

\arxivtrue

\ifone 
\documentclass[journal,12pt,onecolumn,draftclsnofoot,]{IEEEtran}
\else 
\documentclass[journal,10pt]{IEEEtran}
\fi 

\ifarxiv
\usepackage{hyperref}
\else
\fi

\usepackage{acro}
\acsetup{first-style=all}

\DeclareAcronym{URLLC}{
  short = URLLC ,
  long  = Ultra-Reliable Low Latency Communication ,
  class = abbrev
}

\DeclareAcronym{DMRS}{
  short = DMRS ,
  long  = Demodulation Reference Signal,
  class = abbrev
}

\DeclareAcronym{WI}{
  short = WI ,
  long  = Work Item,
  class = abbrev
}

\DeclareAcronym{QAM}{
  short = QAM ,
  long  = Quadrature Amplitude Modulation ,
  class = abbrev
}

\DeclareAcronym{OFDM}{
  short = OFDM ,
  long  = Orthogonal Frequency Division Multiplexing ,
  class = abbrev
}

\DeclareAcronym{OFDMA}{
  short = OFDMA ,
  long  = Orthogonal Frequency Division Multiplexing Access ,
  class = abbrev
}

\DeclareAcronym{E2E}{
  short = E2E ,
  long  = End-to-End ,
  class = abbrev
}

\DeclareAcronym{DL}{
  short = DL ,
  long  = DownLink ,
  class = abbrev
}

\DeclareAcronym{MCS}{
  short = MCS ,
  long  = Modulation and Coding Scheme ,
  class = abbrev
}

\DeclareAcronym{CFI}{
  short = CFI ,
  long  = Control Format Indicator ,
  class = abbrev
}

\DeclareAcronym{UL}{
  short = UL ,
  long  = UpLink ,
  class = abbrev
}

\DeclareAcronym{SR}{
  short = SR ,
  long  = Scheduling Request ,
  class = abbrev
}

\DeclareAcronym{BER}{
  short = BER ,
  long  = Bit Error Rate ,
  class = abbrev
}

\DeclareAcronym{BLER}{
  short = BLER ,
  long  = Block Error Rate ,
  class = abbrev
}

\DeclareAcronym{CRC}{
  short = CRC ,
  long  = Cyclic Redundancy Check ,
  class = abbrev
}

\DeclareAcronym{CB}{
  short = CB ,
  long  = Code Block ,
  class = abbrev
}

\DeclareAcronym{CBG}{
  short = CBG ,
  long  = Code Block Group ,
  class = abbrev
}

\DeclareAcronym{R-CBG}{
  short = R-CBG ,
  long  = Reduced Code Block Group ,
  class = abbrev
}

\DeclareAcronym{AR-CBG}{
  short = AR-CBG ,
  long  = Adaptive Reduced Code Block Group ,
  class = abbrev
}

\DeclareAcronym{TB}{
  short = TB ,
  long  = Transport Block ,
  class = abbrev
}

\DeclareAcronym{BG2}{
  short = BG2 ,
  long  = Base Graph 2 ,
  class = abbrev
}

\DeclareAcronym{BG1}{
  short = BG1 ,
  long  = Base Graph 1 ,
  class = abbrev
}

\DeclareAcronym{SNR}{
  short = SNR ,
  long  = Signal-to-Noise Ratio ,
  class = abbrev
}

\DeclareAcronym{CCE}{
  short = CCE ,
  long  = Control Channel Element ,
  class = abbrev
}

\DeclareAcronym{PDCCH}{
  short = PDCCH ,
  long  = Physical Downlink Control Channel ,
  class = abbrev
}

\DeclareAcronym{PUCCH}{
  short = PUCCH ,
  long  = Physical Uplink Control Channel ,
  class = abbrev
}

\DeclareAcronym{LTE}{
  short = LTE ,
  long  = Long Term Evolution ,
  class = abbrev
}

\DeclareAcronym{NGMN}{
  short = NGMN ,
  long  = Next Generation Mobile Networks ,
  class = abbrev
}

\DeclareAcronym{RNTI}{
  short = RNTI ,
  long  = Radio Network Temporary Identifier ,
  class = abbrev
}

\DeclareAcronym{3GPP}{
  short = 3GPP ,
  long  = 3rd Generation Partnership Project ,
  class = abbrev
}

\DeclareAcronym{HRLLC}{
  short = HRLLC ,
  long  = High-Reliable Low Latency Communication ,
  class = abbrev
}

\DeclareAcronym{TTI}{
  short = TTI ,
  long  = Transmission Time Interval ,
  class = abbrev
}

\DeclareAcronym{sTTI}{
  short = sTTI ,
  long  = short Transmission Time Interval ,
  class = abbrev
}

\DeclareAcronym{RTT}{
  short = RTT ,
  long  = Round Trip Time ,
  class = abbrev
}

\DeclareAcronym{LDPC}{
  short = LDPC ,
  long  = Low-Density Parity-Check ,
  class = abbrev
}

\DeclareAcronym{UE}{
  short = UE ,
  long  = User Equipment ,
  class = abbrev
}

\DeclareAcronym{BS}{
  short = BS ,
  long  = Base Station ,
  class = abbrev
}

\DeclareAcronym{FPR}{
  short = FPR ,
  long  = False-Positive Rate ,
  class = abbrev
}

\DeclareAcronym{FNR}{
  short = FNR ,
  long  = False-Negative Rate ,
  class = abbrev
}

\DeclareAcronym{DCI}{
  short = DCI ,
  long  = Downlink Control Information ,
  class = abbrev
}

\DeclareAcronym{HARQ}{
  short = HARQ ,
  long  = Hybrid Automatic Repeat reQuest ,
  class = abbrev
}

\DeclareAcronym{E-HARQ}{
  short = E-HARQ ,
  long  = Early HARQ ,
  class = abbrev
}

\DeclareAcronym{mMTC}{
  short = mMTC ,
  long  = massive Machine Type Communications ,
  class = abbrev
}

\DeclareAcronym{5G}{
  short = 5G ,
  long  = Fifth Generation ,
  class = abbrev
}

\DeclareAcronym{SPS}{
  short = SPS ,
  long  = Semi-Persistent Scheduling ,
  class = abbrev
}

\DeclareAcronym{PI}{
  short = PI ,
  long  = Pre-emption Indication ,
  class = abbrev
}

\DeclareAcronym{NR}{
  short = NR ,
  long  = New Radio ,
  class = abbrev
}

\DeclareAcronym{eMBB}{
  short = eMBB ,
  long  = enhanced Mobile BroadBand ,
  class = abbrev
}

\DeclareAcronym{LLR}{
  short = LLR ,
  long  = Log-Likelihood Ratio ,
  class = abbrev
}

\DeclareAcronym{VNR}{
  short = VNR ,
  long  = Variable Node Reliability ,
  class = abbrev
}

\DeclareAcronym{BSC}{
  short = BSC ,
  long  = Binary Symmetric Channel ,
  class = abbrev
}

\DeclareAcronym{angelsperarea}{
  short = \ensuremath{a} ,
  long  = The number of angels per unit area ,
  sort  = a ,
  class = nomencl
}
\DeclareAcronym{numofangels}{
  short = \ensuremath{N} ,
  long  = The number of angels per needle point ,
  sort  = N ,
  class = nomencl
}
\DeclareAcronym{areaofneedle}{
  short = \ensuremath{A} ,
  long  = The area of the needle point ,
  sort  = A ,
  class = nomencl
}
\usepackage{units}
\usepackage{cite}
\usepackage{amsmath,amssymb,amsfonts,amsthm}
\usepackage{bbm}
\usepackage{graphicx}
\usepackage{subfigure}
\usepackage{color}
\usepackage{comment}
\usepackage{multirow}
\usepackage{hhline}
\usepackage[normalem]{ulem} 
\newcommand{\stkout}[1]{\ifmmode\text{\sout{\ensuremath{#1}}}\else\sout{#1}\fi}
\ifdraft
\newcommand{\added}[1]{\textcolor{blue}{#1}}
\newcommand{\deleted}[1]{\textcolor{blue}{\stkout{#1}}}
\newcommand{\replaced}[2]{\textcolor{blue}{#1 \stkout{#2}}}
\newcommand{\deletedfloat}[1]{\textcolor{blue}{#1}}
\else
\newcommand{\added}[1]{#1}
\newcommand{\deleted}[1]{}
\newcommand{\replaced}[2]{#1}
\newcommand{\deletedfloat}[1]{}
\fi

\newtheorem{lemma}{Lemma}

\graphicspath{{./figures/}}

\def\BibTeX{{\rm B\kern-.05em{\sc i\kern-.025em b}\kern-.08em
    T\kern-.1667em\lower.7ex\hbox{E}\kern-.125emX}}

\begin{document}
\title{Enhanced Machine Learning Techniques for Early HARQ Feedback Prediction in 5G}

\author{Nils Strodthoff*, Bar{\i}{\c s} G{\" o}ktepe*, Thomas Schierl, Cornelius Hellge and Wojciech Samek
\thanks{Manuscript received Dec., 15th, 2018; revised May, 5th, 2019; accepted May, 29th.}
\thanks{This work was supported by the Bundesministerium f\"ur Bildung und Forschung
(BMBF) through the Berlin Big Data Center under Grant 01IS14013A and the Berlin
Center for Machine Learning under Grant 01IS18037I.}
\thanks{All authors are with Fraunhofer Heinrich Hertz Institute, 10587 Berlin, Germany (e-mail: firstname.lastname@hhi.fraunhofer.de).\newline*Both authors contributed equally to this work (Corresponding authors: Cornelius Hellge, Wojciech Samek).\ifarxiv\newline \copyright 2019 IEEE. Personal use of this material is permitted. Permission from IEEE must be obtained for all other uses, in any current or future media, including reprinting/republishing this material for advertising or promotional purposes, creating new collective works, for resale or redistribution to servers or lists, or reuse of any copyrighted component of this work in other works.\else\fi
}
}

\maketitle

\begin{abstract}
We investigate Early Hybrid Automatic Repeat reQuest (E-HARQ) feedback schemes enhanced by machine learning techniques as a path towards ultra-reliable 
and low-latency communication (URLLC). To this end, we propose machine learning methods to predict the outcome of the decoding process ahead of the end of the transmission.
We discuss different input features and classification algorithms ranging from traditional methods to newly developed supervised autoencoders. These methods are evaluated based on their prospects of complying with the URLLC requirements of effective block error rates below $10^{-5}$ at small latency overheads. We provide 
realistic performance estimates in a system model incorporating scheduling effects to demonstrate the feasibility of E-HARQ across
different signal-to-noise ratios, subcode lengths, channel conditions and system loads, and show the benefit over regular HARQ and existing E-HARQ schemes without machine learning.
\end{abstract}

\begin{IEEEkeywords}
  5G mobile communication, Low latency communication, Physical layer, Machine learning, Anomaly detection, Deep learning
  \end{IEEEkeywords}
\section{Introduction}
The next generation \ac{5G} wireless mobile networks is driven by new emerging use cases, such as \ac{URLLC} \cite{ngmn_whitepaper}. \ac{URLLC} applications such as tactile Internet, industrial automation or smart grids contribute to increasing demands on the underlying communication system which have not existed as such before \cite{vtc_urllc}. Depending on the actual application either very low latency or high reliability or a combination of both are required. In contrast to \ac{LTE}, where services were provided in a best effort manner, 5G networks have to guarantee these requirements. In particular for \ac{URLLC}, the ITU proposed an end-to-end latency of \unit[1]{ms} and a packet error rate of $10^{-5}$ \cite{itu_2020}. These demanding requirements have kicked-off discussions in the 3GPP Rel.\ 16 standardization process on how to fulfill these. Self-contained subframes and grant-free access have been proposed to address these requirements on the air interface side \cite{takeda2017latency}. However, the impact on well-known mechanisms in wireless mobile networks is still unclear. In particular, the \ac{HARQ} procedure poses a bottleneck for achieving \replaced{the previously mentioned}{aforementioned} latencies. 

\ac{HARQ} is a physical layer mechanism that employs feedback to transmit at higher target \acp{BLER}, while achieving robustness of the transmission by providing retransmissions based on the feedback (ACK - acknowledgment / NACK - non-acknowledgment). However, it imposes an additional delay on the transmission, designated as \ac{HARQ} \ac{RTT}. The \ac{HARQ} \ac{RTT} incorporates unavoidable physical delays, such as processing times (hardware delays), propagation delays, and \acp{TTI}, i.e.\ the transmission duration. This lead to the abandonment of \ac{HARQ} for the \unit[1]{ms} end-to-end latency use case of \ac{URLLC} at least for the initial \ac{URLLC} specification in Rel.~15\cite{3gpp_92b_minutes}. This decision implied that the code rate had to be lowered such that a single shot transmission, i.e. no retransmissions and no feedback, achieves the required \ac{BLER}. On the one hand, this simplifies the system design, however on the other hand it sacrifices the overall spectral efficiency of \ac{URLLC} transmissions. Hence, reducing the \ac{RTT} to a minimum for enabling \ac{HARQ} in \ac{URLLC} becomes a critical issue. Even use cases allowing regular \ac{HARQ} retransmissions within the latency budget, profit from the reduced \ac{HARQ} \ac{RTT} enabling more \ac{HARQ} iterations, i.e.\ using a higher code rate for each re-/transmission due to more \ac{HARQ} iterations in total within the latency budget. This improves the overall spectral efficiency in common \ac{HARQ} applications.

In this work, we focus on providing the \ac{HARQ} feedback earlier by predicting the decoding outcome, designated as \ac{E-HARQ} in the following. As will be discussed in detail below, the HARQ RTT is composed of several components including propagation delays and processing delays. Our approach focuses solely on reducing the delay on the receiver side from the start of reception until the feedback is sent. Hence, all other components contributing to the HARQ RTT will be considered fixed for this purpose. \ac{E-HARQ} schemes \cite{early_harq_schemes2,early_harq_schemes} provide the feedback on the decodability of the received signal ahead of the end of the actual transmission process. The crucial component in this setting is the classification algorithm that provides the feedback, which we aim to optimize using machine learning techniques. 

Earlier approaches addressing the feedback prediction problem with the sole exception of \cite{mlharq_workshop} focused exclusively on one-dimensional input features
as \ac{BER} estimates in combination with hard thresholding as classification algorithms \cite{early_harq_schemes2,early_harq_schemes}. In \cite{ldpc_subcodes}, authors introduced the so-called \ac{VNR}, as it will be discussed in Section~\ref{sec:features} below, to exploit the substructures of \ac{LDPC} codes for prediction. However, only a single feature, i.e.\ a single decoder iteration, in combination with hard thresholding has been used. We expect improvements in prediction accuracy 
by extensions in several directions in combination with more complex classification algorithms: (a) the evolution of input features through several decoder iterations considered for the first time in \cite{mlharq_workshop}, (b) higher-dimensional intra-message features 
that in the ideal case leverage knowledge about the underlying block code and (c) history features that leverage information
about the channel state from past submissions that is available at the receiver.

Here we significantly expand the approach put forward in \cite{mlharq_workshop}, where we discuss first E-HARQ results empowered by machine learning techniques. The training and testing is performed on simulated data obtained from stochastic channel models, which are widely used for performance evaluations of physical layer techniques \cite{3gpp_channel_models}. We present an extended theoretical discussion in particular including the extension to multiple retransmissions and a system model that incorporates scheduling effects for the system evaluation thereby allowing a much more precise evaluation of the performance of \ac{E-HARQ}-systems in realistic environments. On the classification side, this is supplemented by extended experiments including different input features and classification algorithms such as a newly developed supervised autoencoder for a larger range of SNR conditions, subcode lengths and different channel models.

The paper is organized as follows: In Section~\ref{sec:eharq} we review the E-HARQ feedback process and investigate 
the role of the classification algorithm in a simple probabilistic model and in a more realistic setting of limited 
system resources. In Section~\ref{sec:ml} we discuss machine learning approaches for the classification problem introducing 
different input features and algorithms. The classification performance as well as the system performance 
is evaluated in Section~\ref{sec:results} for different signal-to-noise ratios,
subcode lengths and channel conditions. We summarize and conclude in Section~\ref{sec:summary}.

\section{Early HARQ Feedback}\label{sec:eharq}
As discussed in the Introduction, \ac{E-HARQ} approaches aim to reduce the \ac{HARQ} \ac{RTT} by providing the feedback on the decodability of the received signal at an earlier stage. This enables the original transmitter to react faster to the current channel situation and to provide additional redundancy at an earlier point. In regular \ac{HARQ}, the feedback generation is strongly coupled to the decoding process. In particular, the receiver applies the decoder on the whole signal representing the total codeword. An embedded \ac{CRC} enables to check the integrity of the decoded bit stream. The result of this check is transmitted back as \ac{HARQ} feedback, either acknowledging correct reception (ACK) or asking for further redundancy (NACK). Providing early feedback (\ac{E-HARQ}) implies decoupling the feedback generation from the decoding process, which introduces a misprediction probability since the actual outcome is not known \replaced{in advance}{afore}. Although misprediction errors are not avoidable, the design choices for the prediction affect the system performance a lot, i.e. asking for more retransmissions than actually required (over-provisioning) or less (under-provisioning). The impact of this design choice is evaluated more in detail in Section~\ref{sec:finitesystemsize}. However, by taking this step, it is possible to use only a portion of the transmission and thus reducing the time from the start of the initial reception to transmitting the feedback ($T_1$). In total, the retransmission is scheduled earlier, hence also reducing the \ac{HARQ} \ac{RTT}, see Fig.~\ref{fig:harq_timeline}. The time for transmitting the feedback and receiving the retransmission ($T_2$) is not affected by this. For \ac{LDPC} codes, \ac{E-HARQ} can be realized under exploitation of the underlying code structure by investigating the feedback prediction problem on the basis of so-called subcodes \cite{ldpc_subcodes,early_harq_tdoc} from the parity-check matrix. These subcodes are constructed by choosing a subset of rows from the original parity-check matrix and all the associated columns, so-called variable nodes \cite{ldpc_subcodes}. The fraction of the subcode length to the full codelength with typical values ranging from 1/2 to 5/6, is designated as subTTI in Fig.~\ref{fig:harq_timeline}. Shorter subcode lengths reduce the \ac{RTT} but at the same time render the prediction problem more complicated.

\begin{figure}[!ht]
  \centering
  \ifone
  \includegraphics[width=.6\columnwidth]{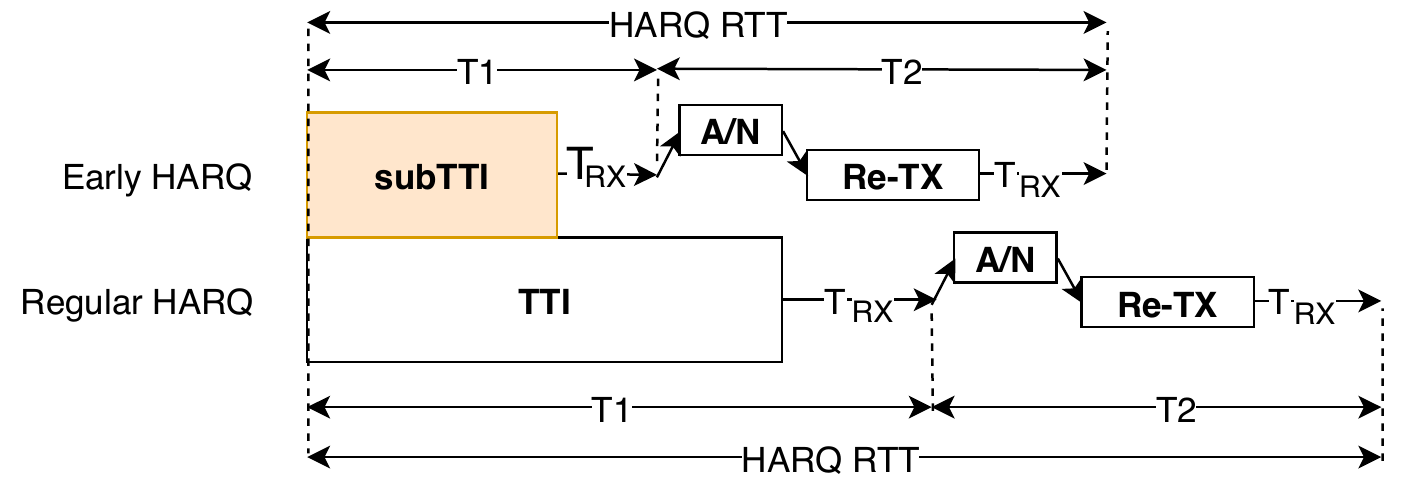}
  \else
  \includegraphics[width=.95\columnwidth]{harq_timeline2}
  \fi
  \caption{Timeline of regular HARQ compared to early HARQ.(HARQ RTT: HARQ round trip time; TTI: transmission time interval; $T_\text{RX}$: processing time at the receiver; A/N: ACK/NACK feedback transmission; Re-TX: retransmission; $T_1$: time from initial reception to feedback transmission $T_2$: time from transmission of feedback to the end of the processing of the retransmission at the receiver)}
  \label{fig:harq_timeline}
\end{figure}

In this section, we first introduce a simple probabilistic system model in Section~\ref{sec:probmodel} to provide an easy tool that evaluates the performance of the here presented \ac{E-HARQ} schemes. However, this model only provides a measure in means of the final \ac{BLER} and additionally implies the assumption of infinite resources by not penalizing unnecessary retransmissions. Hence, in Section~\ref{sec:systemmodel}, we provide a more realistic system model together with the analysis of implications of finite size systems in Section~\ref{sec:finitesystemsize}. This model provides a more suitable tool to evaluate the performance in practical systems, such as \ac{5G} and \ac{LTE}.
The finite-size system argument establishes an optimal point of operation for the \ac{E-HARQ} schemes that is specific for the available system resources and does not exist in a system with unlimited resources.

\subsection{Probabilistic model for single-retransmission \ac{E-HARQ}}
\label{sec:probmodel}

We analyze single-retransmission \ac{E-HARQ} in a simple probabilistic model. The straightforward extension to multiple retransmissions is discussed in 
Appendix~\ref{app:eharqmult}.

\begin{figure}[!ht]
  \centering
  \ifone
  \includegraphics[width=.35\columnwidth]{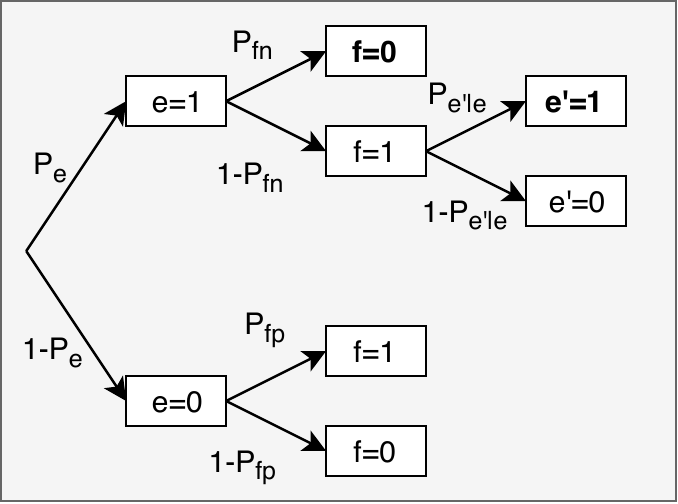}
  \else
  \includegraphics[width=.55\columnwidth]{tree_simple}
  \fi
  \caption{Probabilistic model for single-retransmission E-HARQ (terminal nodes in bold face lead to an effective block error). We use binary random variables $e/e'$ to reflect the state of the transmission ($e/e'=0$: no block error, $e/e'=1$: block error) and a binary random variable $f$ to quantify the feedback sent ($f=0$: \textit{ACK}, $f=1$: \textit{NACK}).}
  \label{fig:tree_left}
\end{figure}

The structure of the probabilistic model for E-HARQ is reflected in Figure~\ref{fig:tree_left}. After the 
initial transmission we end up in a block error state ($e=1$) with probability $P_{\mathrm{e}}\equiv P(e=1)$. Here we follow the common
scheme in imbalanced classification problems encoding the minority i.e.\ block
error class as positive, even though the opposite assignment is often used in the communications literature. In the case $e=0$ the codeword gets decoded correctly
irrespective of the feedback sent and a false positive feedback only implies an
unnecessary transmission, which has no effect on the performance under the infinite resources assumption. We use a binary random variable $f$ to reflect \textit{ACK} $(f=0)$ or \textit{NACK} $(f=1)$ feedback. In the former case we send either \textit{ACK} with
probability $P_{\mathrm{fn}} \equiv P(f=0|e=1)$, which leads to an effective block error, or \textit{NACK} with probability
$P(f=1|e=1)=1-P_{\mathrm{fn}}$. In the latter case the message gets retransmitted which leads to
an effective block error with probability $P_{\mathrm{e'|e}}=P(e'=1|e=1)$.
The value for $P_{\mathrm{e'|e}}$ crucially depends 
on the design of the feedback system most notably on the code rate used for the retransmission. However, one has to keep in mind that a decreased block error rate for the 
retransmission due to a decreased code rate might lead to latency losses due to the necessity of accommodating longer retransmissions. For identical retransmissions using an independent channel realization we would have $P_{\mathrm{e'|e}}= P_{\mathrm{e}}$ or even $P_{\mathrm{e'|e}}< P_{\mathrm{e}}$ if the decoder makes use of information from both transmissions for example using chase combining. For later reference we also define the joint probability $P_{\mathrm{e \wedge e'}}\equiv P_{\mathrm{e}}\cdot P_{\mathrm{e'|e}}=P((e=1)\wedge (e'=1))$. This simple argument leads to an
effective block error probability
\ifone
\begin{equation}
\label{eq:peff}
p_{\mathrm{BLE,eff,1}} = P_{\mathrm{e}}\cdot\left(P_{\mathrm{fn}} + (1-P_{\mathrm{fn}}) P_{\mathrm{e'|e}} \right)= P_{\mathrm{e}} P_{\mathrm{fn}}+ P_{\mathrm{e \wedge e'}} (1-P_{\mathrm{fn}})\,.
\end{equation}
\else
\begin{equation}
\begin{split}
\label{eq:peff}
p_{\mathrm{BLE,eff,1}} &= P_{\mathrm{e}}\cdot\left(P_{\mathrm{fn}} + (1-P_{\mathrm{fn}}) P_{\mathrm{e'|e}} \right)\\
&= P_{\mathrm{e}} P_{\mathrm{fn}}+ P_{\mathrm{e \wedge e'}} (1-P_{\mathrm{fn}})\,.
\end{split}
\end{equation}
\fi

The effect of an imperfect feedback channel could be easily incorporated in this formalism by defining effective false negative/positive rates but is omitted here for simplicity.
Empirically we can replace $P_{\mathrm{e}}$ and $P_{\mathrm{e \wedge e'}}$ by estimated block error rates and the
conditional probability $P_{\mathrm{fn}}$ by the classifier's false negative rate (FNR) as obtained from the confusion matrix. Obviously the lowest possible effective BLER is achieved for
perfect feedback, i.e.\ $P_{\mathrm{fn}}=0$, for which we have
$p_{\mathrm{BLE,eff,1}}=P_{\mathrm{e \wedge e'}}$. (\ref{eq:peff}) only depends on the baseline BLERs $P_{\mathrm{e}}$ and $P_{\mathrm{e \wedge e'}}$ and 
the classifier's false negative rate $P_{\mathrm{fn}}$ with leading order contribution given by $P_{\mathrm{fn}}\cdot P_{\mathrm{e}}$.
In the limit where the $P_{\mathrm{fn}} \ll P_{\mathrm{fb,e}}$ the leading behavior is just $P_{\mathrm{e}}\cdot P_{\mathrm{fb,e}}$ and hence independent of 
the classification performance.

Considering the question of latency, the simplest metric is to consider the
expected number of retransmissions $\langle \Delta T_1 \rangle$, which relates to the spectral efficiency of the approach. Therefore we evaluate 
the probability $P_{\mathrm{r,1}}$ for a single retransmission. Again using 
Figure~\ref{fig:tree_left}, we obtain
\begin{equation}  
P_{\mathrm{r,1}}\equiv P_{\mathrm{r}} =P_{\mathrm{e}}(1-P_{\mathrm{fn}}) + (1-P_{\mathrm{e}}) P_{\mathrm{fp}}\,.
\label{eq:pr1}
\end{equation} 
As above, the conditional probability $P_{\mathrm{fp}}\equiv P(f=1|e=0)$ can be identified empirically with the classifier's false positive rate (FPR). 
The leading order contribution to (\ref{eq:ETr}) is given by $P_{\mathrm{e}}+P_{\mathrm{fp}}$ and the number of expected retransmissions 
therefore profits from a decreased FPR. For the case of a single retransmission, the expected number of retransmissions $\langle \Delta T_1 \rangle$
coincides with the single-retransmission probability, 
\begin{equation}
\label{eq:ETr}
\langle \Delta T_1 \rangle= P_{\mathrm{r,1}}\,.
\end{equation}
These results already hint at the crucial importance of adjusting the classifier's working point by balancing FNR versus FPR: A reduction of the FNR leads to a smaller effective block error probability, see (\ref{eq:peff}), but 
comes along with an increased FPR, as the two kinds of classification errors counterbalance each other. This in turn leads 
to an increase in latency, see (\ref{eq:ETr}). From the present discussion it might seem a reasonable strategy to target an arbitrarily small FNR such that the 
effective block error probability approaches the theoretical limit. However, this argument only holds for a system with unlimited resources,
as will be discussed below.
\subsection{System model}
\label{sec:systemmodel}
In order to derive a tool for evaluation of the performance of the discussed predictors,  we introduce in this section a more sophisticated system model that leans on the structure of today's mobile network technologies. In cellular networks, such as \ac{LTE} and \ac{5G}, \ac{OFDMA} has been established due to its scheduling flexibility. Especially, opportunistic scheduling allows to use the best possible channel for a transmission. Here, we assume a simplified \ac{OFDMA} system with equally sized $N_{\mathrm{res}}$ resources, i.e.\ frequency resources and a defined transmission duration in time, so-called \ac{TTI}. The \ac{HARQ} mechanism, regular \ac{HARQ} as well as \ac{E-HARQ}, requests based on the received parts of the transmission a retransmission, which is scheduled at earliest after $T_{\mathrm{RTT}}$ time slots. 

The main advantage of \ac{E-HARQ} over regular \ac{HARQ} is the reduced \ac{HARQ} \ac{RTT}. Hence, depending on the latency budget more \ac{HARQ} iterations, i.e. more re-/transmissions each incorporating feedback from the receiver within the latency constraint, might be used to improve the system performance. In the following, we assume the \ac{TTI} length the default time unit. In practical systems, there exist several possibilities to design the \ac{TTI} length such that it fits the requirements. The number of \ac{OFDM} symbols belonging to a \ac{TTI} or the subcarrier spacing of the \ac{OFDM} modulation can be  changed, such that the physical time duration of a \ac{TTI} scales accordingly. In this work, we evaluated two different system approaches, long and short \ac{TTI} lengths.

The \ac{HARQ} \ac{RTT} is mainly comprised by the processing time, which scales with the \ac{TTI} length in general \cite{takeda2017latency}, and the time required for transmitting the feedback, which does not depend on the \ac{TTI} length. Thus, for long \ac{TTI} lengths this time can be considered relatively small. However, for short \ac{TTI} lengths this constant component has to be considered for \ac{E-HARQ} as well as for regular \ac{HARQ} systems. Hence, for long \acp{TTI}, we assumed $T_{\mathrm{RTT}} = 1$ for rate-1/2 \ac{E-HARQ}, which means that the retransmission is received in the next TTI, and $T_{\mathrm{RTT}} = 2$ for regular \ac{HARQ}, so that for regular \ac{HARQ} one TTI has to be skipped. Analogously, for short \acp{TTI}, $T_{\mathrm{RTT}} = 5$ for rate-5/6 \ac{E-HARQ} and $T_{\mathrm{RTT}} = 6$ for regular \ac{HARQ}. 
Depending delay constraint $T_\mathrm{c}$ this results to a maximum number of retransmissions possible within the said latency budget.
For long and short \acp{TTI} this allows depending on the system load up to two retransmissions in the \ac{E-HARQ}-scheme compared to only one in the regular \ac{HARQ}-scheme. Due to the scalability of the \ac{TTI} length, the absolute value of $T_c$ might be set to an arbitrary value, e.g.\, \unit[1]{ms}.
Thanks to the \replaced{previously mentioned}{aforementioned} opportunistic scheduling possibilities of \ac{OFDMA}, we assume that the retransmission is independent of the previous transmission, i.e. $P_{\mathrm{e'|e}} = P_{\mathrm{e}}$ and the total \ac{BLER} $P_{\mathrm{e,total}} = (P_{\mathrm{e}})^{n+1}$, where $n$ is the number of retransmissions. Furthermore, an i.i.d.\ packet arrival rate $P_{\mathrm{A,UE}}$ for each \ac{UE} is assumed. Thus, a single \ac{UE} can only have one new transmission per time slot. The arrival rate for the high load scenario, as given in Table~\ref{tab:sys_param}, has been chosen such that the \ac{E-HARQ} performance saturates in the \ac{URLLC} relevant range. For the medium load scenario, the arrival rate has been chosen slightly lower such that no saturation behavior is observable in the relevant range. For simplicity the following argument is carried for a perfect feedback channel, i.e.\ for $P_{\mathrm{fb,e}}=0$, which is a reasonable assumption considering the results of the previous implying that the feedback error probability is at most of subleading importance. The system parameters are summarized in Table~\ref{tab:sys_param}.
\ifone
\begin{table}[ht]
\centering
\caption{System Evaluation Parameters}
\label{tab:sys_param}
\begin{tabular}{l|l}
Base time unit&\ac{TTI} (e.g., for \unit[1]{ms}: long \ac{TTI} - $1 \text{TTI} \hat{=} \frac{1}{3} \text{ms}$, short \ac{TTI} - $1 \text{TTI} \hat{=} \frac{1}{11} \text{ms}$) \\
\hline
\ac{UE} packet arrival rate - $P_{\mathrm{A,UE}}$&medium load - 0.3, high load - 0.36\\
\hline
Number of \acp{UE} - $N_{\mathrm{UE}}$&20\\
\hline
Number of resources per time slot - $N_{\mathrm{res}}$ &10\\
\hline
Delay constraint - $T_{\mathrm{c}}$&long \ac{TTI} - 3, short \ac{TTI} - 11\\
\hline
long \ac{TTI} \ac{HARQ} \ac{RTT} - $T_{\mathrm{RTT}}$&1 (\ac{E-HARQ} 1/2), 2 (regular \ac{HARQ})\\
\hline
short \ac{TTI} \ac{HARQ} \ac{RTT} - $T_{\mathrm{RTT}}$&5 (\ac{E-HARQ} 5/6), 6 (regular \ac{HARQ})\\
\hline
\ac{BLER} of (re-)transmissions - $P_{\mathrm{e}}$& as in Table~\ref{tab:basic_vnr}\\
\end{tabular}
\end{table}
\else
\begin{table}[ht]
\centering
\label{tab:sys_param}
\caption{System Evaluation Parameters}
\begin{tabular}{l|l}
\ac{UE} packet arrival rate - $P_{\mathrm{A,UE}}$&medium load - 0.3,\\&high load - 0.36\\
\hline
Number of \acp{UE} - $N_{\mathrm{UE}}$&20\\
\hline
Number of resources &10\\
per time slot - $N_{\mathrm{res}}$&\\
\hline
Delay constraint - $T_{\mathrm{c}}$&long symbols - 3,\\&short symbols - 11\\
\hline
long \ac{TTI} \ac{HARQ} \ac{RTT} - $T_{\mathrm{RTT}}$&1 (\ac{E-HARQ} 1/2),\\
&2 (regular \ac{HARQ})\\
\hline
short \ac{TTI} \ac{HARQ} \ac{RTT} - $T_{\mathrm{RTT}}$&5 (\ac{E-HARQ} 5/6),\\
&6 (regular \ac{HARQ})\\
\hline
\ac{BLER} of (re-)transmissions - $P_{\mathrm{e}}$& as given in Table~\ref{tab:basic_vnr}\\
\end{tabular}
\end{table}
\fi
\subsection{Implications of finite system size}
\label{sec:finitesystemsize}
In practical systems, there is a trade-off between the \ac{FNR} and \ac{FPR} due to the limited amount of available resources. Whereas a lower \ac{FNR} increases the effective \ac{BLER}, as shown in the Section~\ref{sec:probmodel}, it increases the transmission overhead on the other hand. Depending on the available resources this leads to resource shortage, also causing additional delays since transmissions cannot be scheduled in the designated time slots. This brings us to the term of packet failure rate which is described by the probability that a packet is delivered successfully within a given delay constraint, i.e.\ latency budget \added{ $T_\mathrm{c}$}. Interestingly, there is an optimal operation point which captures the trade-off such that the packet failure rate is minimized. 

For the assumptions on the system model described in the previous section, the packet failure probability is given as
\begin{equation}\label{eq:pfr}
P_{\mathrm{pf}} = (1-P_{\mathrm{\added{S}},0}) + P_{\mathrm{\added{S}},0} P_{\mathrm{e}}P_{\mathrm{H,e},1}\,,
\end{equation}
where $P_{\mathrm{S},j} \equiv P(T_{j} \leq T_\mathrm{c})$ denotes the probability of scheduling $j$ transmissions within the time constraint $T_c$ and $P_{\mathrm{H,e},j}$ denotes the failure probability after the $j$th transmission. For a single retransmission ($n=1$) the latter is given by $P_{\mathrm{H,e},1}=P_{\mathrm{fn}}+(1-P_{\mathrm{fn}})[(1-\tfrac{{P_{\mathrm{S},1}}}{{P_{\mathrm{S},0}}})+\tfrac{P_{\mathrm{S},1}}{P_{\mathrm{S},0}} P_e]$. Inserting this expression into (\ref{eq:pfr}) leads to the familiar form (\ref{eq:peff}) up to scheduling probabilities $P_{\mathrm{S},j}$. Generalizing to multiple retransmissions, the error probability for the $j$th resubmission can be defined recursively via 
\begin{equation}
P_{\mathrm{H,e},j} =
P_{\mathrm{fn}} + (1- P_{\mathrm{fn}})\left[\left(1-\tfrac{P_{\mathrm{S},j}}{P_{\mathrm{S},j-1}}\right)+ \tfrac{P_{\mathrm{S},j}}{P_{\mathrm{S},j-1}} P_{\mathrm{e}}P_{\mathrm{H,e},j+1}\right],
\end{equation}
where we set $P_{\mathrm{H,e},j}=1$ if $j$ exceeds the maximum number of retransmissions, i.e.\ if $j>n$. As shown above for a single retransmission, (\ref{eq:pfr}) reduces to (\ref{eq:peff}) if one sets all scheduling probabilities to one. The same applies for multiple retransmissions for which one obtains (\ref{eq:peffn}). Therefore, (\ref{eq:pfr}) can be seen as a generalized version of the effective \ac{BLER}.

However, the effective \ac{BLER} does not consider the finite resources and thus cannot capture the actual performance of the evaluated \ac{HARQ} schemes in a practical implementation. We will refer to this case as the infinite resource baseline compared to the finite resource baselines discussed below.

At first glance, (\ref{eq:pfr}) suggests minimizing the FNR $P_\mathrm{fn}$. However, a closer examination reveals that the scheduling probabilities $P_{\mathrm{S},j}$ carry a dependence on both FNR and FPR via the underlying resource distribution function, see Appendix~\ref{sec:sched}. FNR and FPR counteract each other in the sense that a decreased FNR will lead to an increase in the FPR.
Considering the dependency on the resource distribution function, an increase of the FPR $P_\mathrm{fp}$ increases the load on the system, thus lowers the probability that a transmission and its retransmission is scheduled within the time constraint. This fact is already apparent from the expected number of retransmission as obtained in (\ref{eq:ETr}) which scales with the FPR at leading order. This suggests that the packet failure probability seen as a function of the FNR will show a minimum characterizing an optimal trade-off between FNR and FPR for the given system resources.

The derived packet failure probability $P_\mathrm{pf}$ within a fixed latency budget rather than spectral efficiency represents the most relevant performance metric for practical evaluations. Additionally, apart from comparing the different \ac{E-HARQ} schemes among each other, it enables a performance comparison with regular \ac{HARQ}, which is crucial if \ac{E-HARQ} is considered for \ac{URLLC}. Here, aside the system setup presented in the previous section, for regular \ac{HARQ} the \ac{FNR} and \ac{FPR} is assumed to be zero as false predictions can be neglected due to the \ac{CRC} included in the transmission. 

\section{Machine Learning for early HARQ}
\label{sec:ml}
The machine learning task of predicting the decodability of a message based on information from at most the first few decoder iterations is an inherently imbalanced classification problem.
This imbalance is a direct consequence of the base BLERs of the order $10^{-3}$ that
are required in order to be able to reach effective BLERs of the order $10^{-5}$, see (\ref{eq:peff}).
Different ways of dealing with this imbalance have been explored, see 
\cite{DBLP:journals/corr/BrancoTR15} for a review. These can be categorized as 
cost-sensitive learning, rebalancing techniques and threshold moving. The discussion 
in this section focuses on the latter 
in the sense of readjusting the decision boundary of any trained model that 
outputs probabilities for the predicted classes, see also \cite{collell2016reviving} and references therein. 

By moving the decision boundary one is 
able to investigate the discriminative power of a given classifier over a whole range of different working points. This is typically analyzed in terms of Receiver-Operation curves (ROC) or Precision-Recall (PR) curves.
In order to summarize the classifier's performance with a single number, one conventionally
resorts to reporting area-under-curve (AUC) metrics. Here we focus on the PR curve and the corresponding 
area under the PR curve, AUC-PR, rather than the ROC-curve 
as the former has been shown to better reflect the classifier's performance
for highly skewed datasets \cite{DBLP:journals/corr/abs-1801-03149,Davis:2006:RPR:1143844.1143874}. 
However, when summarizing the discriminative power of a classifier using a single figure, one loses fine-grained information about classification performance at different working points. This is particularly true since the full AUC naturally covers the whole range values for the decision boundary, many of which are irrelevant for practical applications where the classification performance in the small FNR-regime is most relevant. In addition, the actual implementation of the classifier requires a definite choice for the decision threshold. Therefore we supplement the global AUC-PR information with an analysis based on FNR-PPR curves. It is worth noting that the FNR-FPR curves directly relate to ROC curves since the true positive rate TPR that is plotted on the ordinate of the ROC-curve relates to the FNR via TPR = 1 - FNR. FNR and FPR represents the natural choice in our case since they represent the key output figures from the system point of view, see Sec~\ref{sec:probmodel}.

\subsection{Input features}
\label{sec:features}
We distinguish single-transmission-features derived from a single transmission and history information 
from past transmissions. In principle all these features can be combined at will to form the set 
of input features for the classification algorithm. 

The raw data for a single transmission provided by the simulation is given by (a posteriori) LLR values 
after different decoder iterations.
\ac{E-HARQ} approaches to reduce the \ac{HARQ} \ac{RTT} have been first discussed in \cite{early_harq_schemes2} and \cite{early_harq_schemes}. This approach estimates the \ac{BER} based on the \acp{LLR} and utilizes a hard threshold to predict the decodability of the received signal. The \ac{LLR} gives information on the likelihood of a bit being either $1$ or $0$. Denoting $\boldsymbol{y}$ as the observed sequence at the receiver, the LLR of the $k^{th}$ bit $b_k$ is defined as\deleted{:}
\begin{equation}
\label{eq:llr}
L(b_k) = \log\frac{P(b_k = 1 | \boldsymbol{y})}{P(b_k = 0 | \boldsymbol{y})}\,.
\end{equation}
Having the \acp{LLR} of a subcode or the whole codeword allows to calculate an estimated \ac{BER} for the received signal vector via 
\begin{equation}
\label{eq:ber_llr}
\hat{\mathit{BER}} = \frac{1}{M}\sum_k \frac{1}{1 + |L(b_k)|}\, ,
\end{equation}
where $M$ is the length of the \ac{LLR} vector. Based on this metric the decoding outcome is predicted, where a higher $\hat{\mathit{BER}}$ means a lower probability of successful decoding.

A further improved approach has been presented in \cite{ldpc_subcodes} and \cite{early_harq_tdoc}. The authors propose to exploit the code structure to improve the prediction performance. In case of \ac{LDPC} codes, this is realized by constructing so-called subcodes from the parity-check matrix. Using a belief-propagation based decoder on the \acp{LLR} of the subcodeword results in a posteriori \acp{LLR}:
\begin{equation}
\label{eq:ap_llr}
\Lambda^{(j)}_k = \Lambda^{(j-1)}_k  + \sum_{m \in \mathcal{M}(k)} \beta^{(j)}_{m,k} ,
\end{equation}
where $\Lambda^{(0)}_k\equiv L(b_k)$, $\mathcal{M}(k)$ is the set of check nodes which are associated to the variable node of $k$ and $\beta^{(j)}_{m,k}$ is the check-to-variable node message from check node $m$ to variable $k$. Here we use the superscript $j$ in $\Lambda^{(j)}_k$ to denote the decoder iteration after which the posteriori \acp{LLR} were extracted with the obvious identification $\Lambda^{(0)}_k\equiv L(b_k)$. Again, the a posteriori \acp{LLR} are mapped to the same metric for each belief-propagation iteration, designated as \ac{VNR},
\begin{equation}
\label{eq:rv}
\mathit{VNR}_j = \frac{1}{M}\sum_i \frac{1}{1 + |\Lambda^{(j)}_i|}\, ,
\end{equation}
where $M$ is the length of the subcodeword and $j$ denotes the belief-propagation iteration. Hence, $\mathit{VNR}_0$ corresponds to $\hat{BER}$. In \cite{ldpc_subcodes}, the authors used a hard threshold applied $\mathit{VNR}_5$ to predict decodability.

Assuming the receiver is operating on the same channel across different transmissions, it might be possible to increase the prediction performance 
by incorporating information from previous transmissions. This includes all features used as single-transmission features and in addition features that 
are only available after the end of the decoding process. As two representative examples for history features we investigate VNRs from past submissions (VNR\_HIST)
and information about the Euclidean distance between the correct codeword and the received signal vector (EUCD\_HIST). Here one has to 
keep in mind that the latter information is only available if the correct codeword is known to the receiver as for example from a previous pilot transmission but 
strictly speaking it cannot be reliably obtained from an ordinary previous transmission as even a correct CRC does not imply a correctly decoded transmission. For a given set of history features we consider means of the history features under consideration extracted from different numbers of past transmissions using sliding windows of length 1,2,5 and 9 in order to allow the classifier to extract information 
from past channel realizations at different time scales.

\subsection{Classification algorithms}
\label{sec:algorithms}
As discussed in the introduction, we can view the problem either as a heavily imbalanced classification problem or as an anomaly detection problem. Here we briefly discuss suitable algorithms for both of approaches. As examples for binary classification algorithms we consider hard threshold (HT) classifiers,  logistic regression (LR) (with $L_2$ regularization and balanced class weights) and Random Forests (RF). HT applied to $\mathit{VNR}_0$/$\mathit{VNR}_5$-data (referred to as HT0 and HT5 in the following) yield the classifiers used in the literature so far \cite{early_harq_schemes2,ldpc_subcodes}. For anomaly detection \cite{Chandola:2009:ADS:1541880.1541882} one distinguishes unsupervised, semi-supervised and supervised approaches depending on whether only unlabeled examples, only majority-class examples or labeled examples from both classes are available for training. As anomaly detection algorithms we consider Isolation Forests (IF) \cite{liu2008isolation} as classical tree-based semi-supervised anomaly detection algorithm and supervised autoencoder (SAE) as a novel neural-network based approach for supervised anomaly detection, see Appendix~\ref{sec:SAE} for details. All classifiers apart from HT0 and HT5 operate on the first six VNRs $\mathit{VNR}_0,\ldots,\mathit{VNR}_5$ as input features.
We leverage the implementations from scikit-learn \cite{scikit-learn} apart from SAC that was implemented in PyTorch \cite{paszke2017automatic}.

\section{Results}
\label{sec:results}

\subsection{Simulation setup}
\ifone
\begin{table}[t]
\centering
\caption{Link-level simulation assumptions for training and test set generation.}
\label{tab:lls_assump}
\begin{tabular}{l|l}
Transport block size & 360 bits \\
\hline
Channel Code&Rate-1/5 LDPC BG2 with Z = 36, see \cite{5g_channel_coding_spec}\\
\hline
Modulation order and algorithm&QPSK, Approximated LLR\\
\hline
Waveform&3GPP OFDM, 1.4 MHz, normal cyclic-prefix\\
\hline
Channel type&1~Tx 1~Rx, TDL-C 100~n~, 2.9~GHz,\\
&3.0~km/h (pedestrian) or 100.0~km/h (vehicular)\\
\hline
Equalizer&Frequency domain MMSE\\
\hline
Decoder type&Min-Sum\\
\hline
Decoding iterations&50\\
\hline
VNR iterations&5\\
\end{tabular}
\label{tab:SIM}
\end{table}
\else
\begin{table}[t]
\centering
\caption{Link-level simulation assumptions for training and test set generation.}
\label{tab:lls_assump}
\begin{tabular}{l|l}
Transport block size & 360 bits \\
\hline
Channel Code&Rate-1/5 LDPC BG2 with Z = 36,\\
&see \cite{5g_channel_coding_spec}\\
\hline
Modulation order and algorithm&QPSK, Approximated LLR\\
\hline
Waveform&3GPP OFDM, 1.4 MHz,\\
&normal cyclic-prefix\\
\hline
Channel type&1~Tx 1~Rx, TDL-C 100~n~, 2.9~GHz,\\
&3.0~km/h (pedestrian) or\\
&100.0~km/h (vehicular)\\
\hline
Equalizer&Frequency domain MMSE\\
\hline
Decoder type&Min-Sum\\
\hline
Decoding iterations&50\\
\hline
VNR iterations&5\\
\end{tabular}
\label{tab:SIM}
\end{table}
\fi
We compare classification performance of different classifiers based on AUC-PR and
FNR-FPR curves. As external parameters we vary the SNR between 3.0 and 4.0 dB which results to \acp{BLER} in the considered \ac{URLLC} regime taking the \ac{HARQ} retransmissions into account, and subcode lengths between
1/2 and 5/6 to evaluate a rather aggressive prediction versus a more conservative one. The simulation setup used to produce training and test data follows the
one reported in \cite{ldpc_subcodes} and summarized in Table~\ref{tab:lls_assump}.
We use the raw simulation output as well as a number of derived features. Here we consider both 
single-transmission features as well as history-features that incorporate information from a number 
of past transmissions, see Appendix~\ref{sec:features} for a detailed discussion. We then investigate the 
performance of a number of classification algorithms operating on these input features, see Appendix~\ref{sec:algorithms} 
for a detailed breakdown.
In all cases we use 1M transmissions with independent channel realizations for training and evaluate on 
a test set comprising at least 1M transmissions. The size of the test set for each SNR/subcode combination is given in the second column of Table~\ref{tab:basic_vnr}. Hyperparameter tuning is performed once for the pedestrian channel (at SNR 4.0~dB and subcode length 5/6) on an additional validation set also comprising 1M samples. We standard-scale all different sets of input features 
independently using training set statistics. In this way we obtain a reasonable input normalization that is 
required for certain classification algorithms 
while keeping relative difference within different input feature groups intact.

\subsection{Classification performance}
\label{sec:vnrfeatures}

\begin{table*}[ht]
\centering
\caption{Comparing classification performance based on AUC-PR (classifiers as specified in Section~\ref{sec:algorithms}).}
\begin{tabular}{c||c|c||c|c||c|c||c|c}
SNR SC ch&\#train/\#test&BLER & HT0 & HT5& LR & RF & IF & SAE \\
\hline
\hline
4.0dB 5/6 ped &1M/3M& 0.001604 & 0.811 & 0.902 & 0.905 &0.907 & 0.890 & \textbf{0.908}\\
4.0dB 1/2 ped &1M/4M&0.001626 & 0.801 &0.799 & \textbf{0.834} & 0.832 & 0.827 & \textbf{0.834} \\
\hline
3.5dB 5/6 ped &1M/1M& 0.002841 & 0.844& 0.920 & 0.921 & 0.924& 0.912 & \textbf{0.926}\\
3.5dB 1/2 ped &1M/4M& 0.002777 & 0.821 & 0.814&\textbf{0.847} &0.846 & 0.839 & \textbf{0.847}\\
\hline
3.0dB 5/6 ped &1M/1.5M& 0.004742  & 0.863& 0.927 &\textbf{0.934} &\textbf{0.934} & 0.923 & \textbf{0.934}\\
3.0dB 1/2 ped &1M/1.5M& 0.004742 & 0.851 & 0.840  & 0.872 &0.871 & 0.865 &\textbf{0.874}\\
\hline
\hline
3.5dB 1/2 veh &1M/3M& 0.002866  & 0.824& 0.818 &\textbf{0.851} &0.850 & 0.846 &\textbf{0.851}\\
\end{tabular}
\label{tab:basic_vnr}
\end{table*}

We start by discussing the classification performance for different classification algorithms based on $\mathit{VNR}$-features extending the analysis from \cite{mlharq_workshop}. The classification results are compiled in Table~\ref{tab:basic_vnr}. We compare AUC-PR that characterizes the overall discriminative power of the algorithm and which tends to 1 for a perfectly discriminative classifier. The largest improvements to the simplest thresholding method HT0 is seen for longer subcode lengths such as 5/6. In these cases more complex classification methods applied to the full VNR-range show only small improvements over the HT5 threshold baseline. A different picture emerges at smaller subcode lengths. Here using VNRs from higher decoder iterations (HT5) does not improve or even worsen the classification performance compared to HT0. Here more complex classification algorithms show their true strengths and show larger improvements compared to HT0/HT5. This is a plausible result since decreasing the subcode length renders the classification problem more complicated and more complex classifiers can profit more from this complication. If we assess the difficulty of the classification problem based on the scores achieved by the classifiers, a clear picture emerges: As discussed before decreasing the subcode length for fixed SNR renders the classification problem more difficult, whereas decreasing the SNR for fixed subcode length has the opposite effect most notably because of an increasing BLER. On the other hand the BLER sets the baseline for the HARQ performance, see (\ref{eq:peff}), which overcompensates the positive effects of the improved classification performance. The overall best discriminative power across different SNR-values, subcode lengths and channel conditions shows the supervised autoencoder closely followed by regularized logistic regression. The fact that the AUC-PR results for LR, RF and SAE are so close just reflects a similar overall discriminative power of these algorithms despite of fundamentally different underlying principles.

\begin{figure*}[ht]
\centering     
\subfigure[4.0 dB subcode 5/6]{\label{fig:vnr40_56}\includegraphics[width=.45\textwidth]{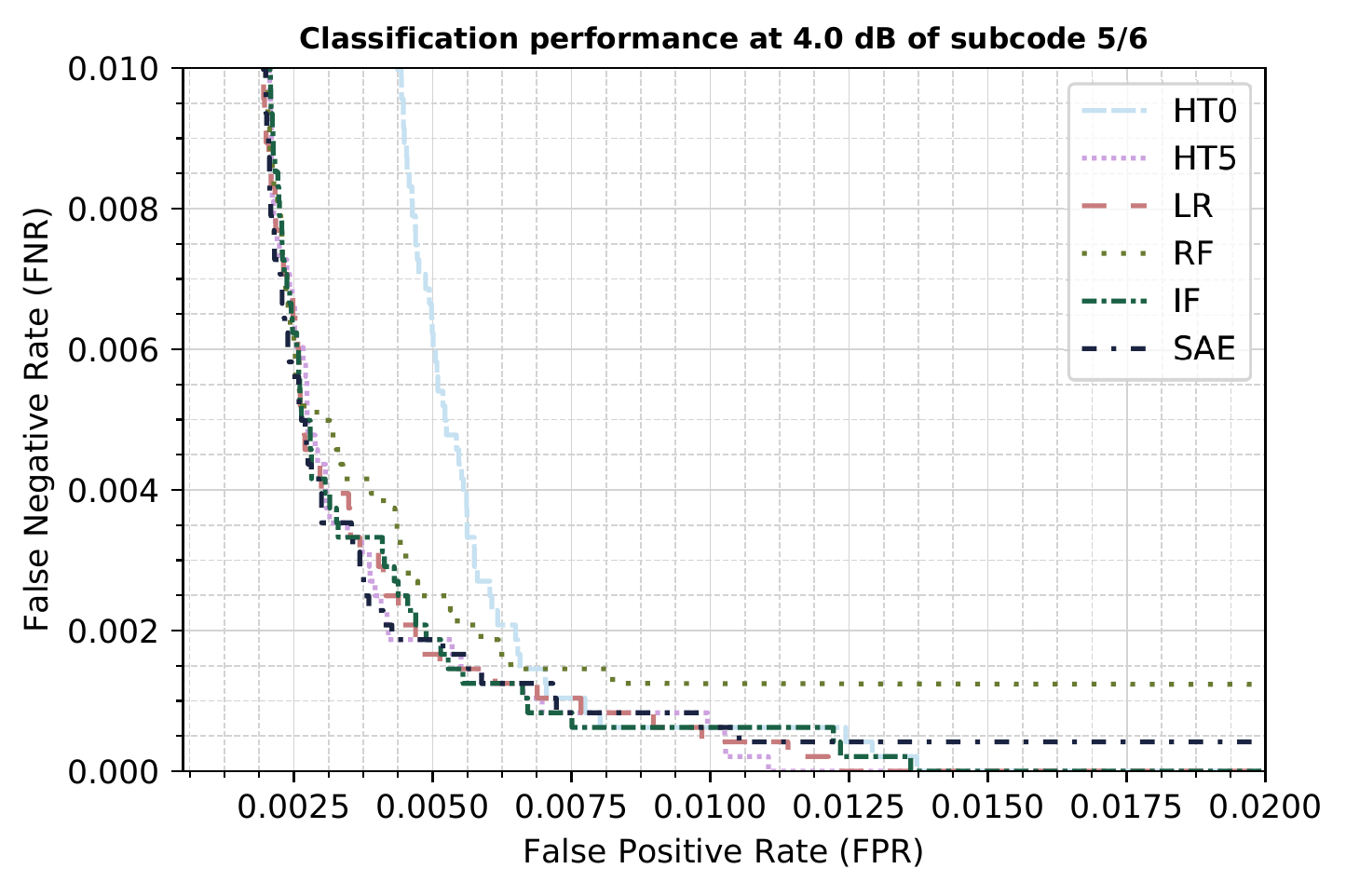}}
\subfigure[4.0 dB subcode 1/2]{\label{fig:vnr40_12}\includegraphics[width=.45\textwidth]{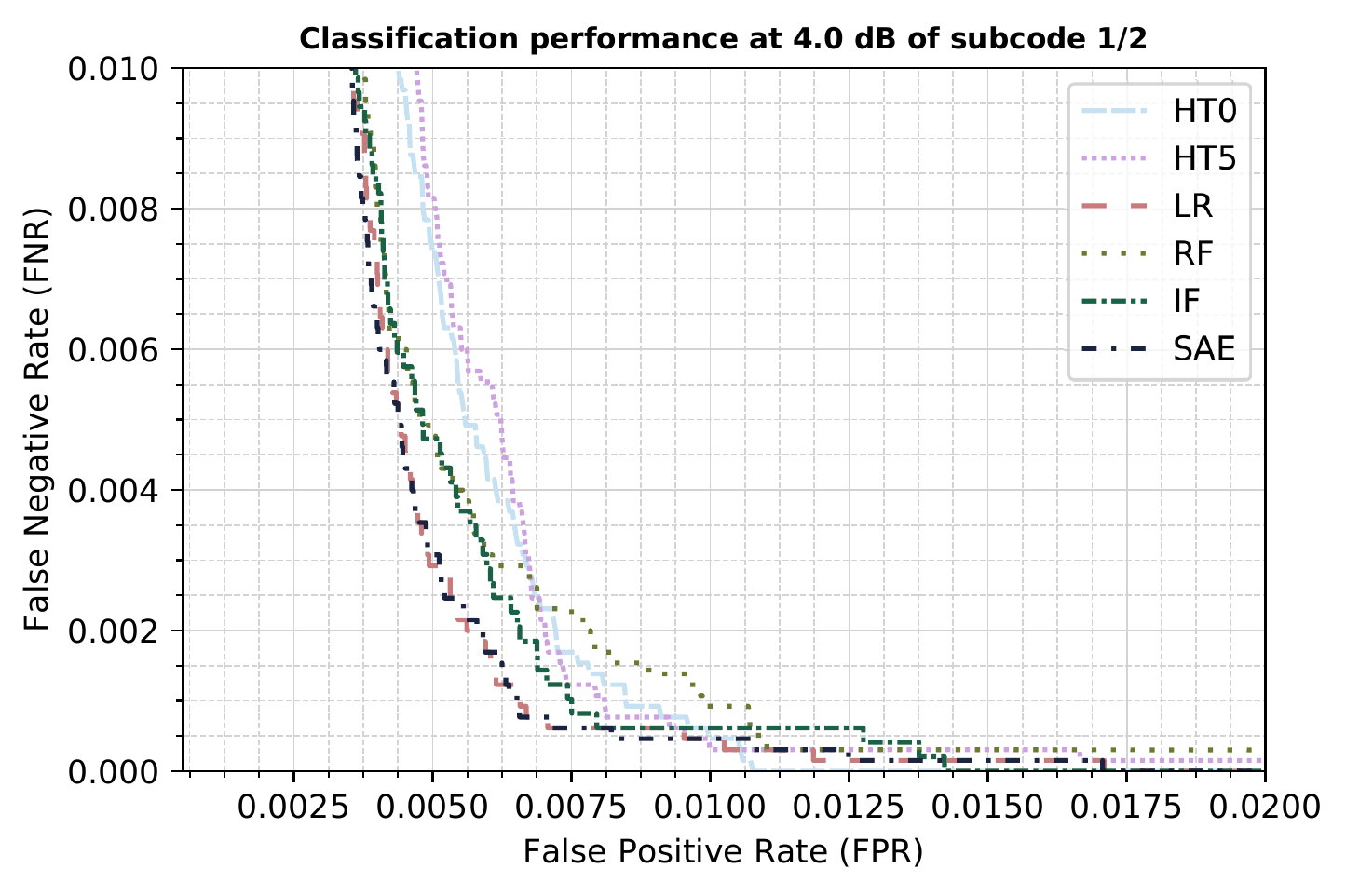}}
\caption{Selected examples for classification performance based on VNR-features in the pedestrian channel.}
\label{fig:basic_vnr}
\end{figure*}

This does, however, not imply coinciding FNR-FPR curves, where the classifiers show rather different behavior in certain FNR regions, see Figure~\ref{fig:basic_vnr} for selected results. Random Forests, for example, show in general a very good overall performance but are considerably weaker than other classifiers in the small FNR-regime. 
When looking at FNR-FPR curves as the ones presented in Figure~\ref{fig:basic_vnr}, one has to keep in mind that it is very difficult in the extremely imbalanced regime to obtain reliable estimates of the FNR as both the numerator (false negatives) and the denominator (sum of false negatives and true positives) are small numbers requiring large sample sizes for a stable evaluation. This applies in particular to the region of small FNRs below 0.001.

To summarize, we clearly demonstrated that incorporating the evolution of the VNR across the first five decoder iterations into more complex classification algorithms such as logistic regression or supervised autoencoders leads to gains in the overall classification performance in particular in comparison to hard threshold baselines. This conclusion holds for various SNR-values, subcode lengths and channel conditions. Implications of these findings for the system performance will be discussed in Section~\ref{sec:systemperformance}.

We restrict the investigation of history features to the SAE classifier as the best-performing classifier from the previous section. However, we checked that the qualitative conclusions about the importance of history features hold irrespective of the classification algorithm under consideration. In Table~\ref{tab:history_features} we discuss the impact of history features on the classification performance in addition to the VNR-features discussed above. 

\begin{table}[ht]
\centering
\tabcolsep 3pt
\caption{Comparing classification performance based on AUC-PR upon including history features (for SAE).}
\begin{tabular}{c||c|c|c}
features& 4.0dB 1/2 ped &  3.5dB 1/2 ped & 3.5dB 1/2 veh\\
\hline
\hline
VNR&0.834&0.847&0.851\\
\hline
VNR+VNR\_HIST&0.860&0.872&0.852\\
VNR+EUCD\_HIST&\textbf{0.883}&\textbf{0.892}&\textbf{0.861}\\
\end{tabular}
\label{tab:history_features}
\end{table}

\begin{figure*}[!ht]
  \centering     
  \subfigure[4.0 dB subcode 5/6]{\label{fig:vnr40_56_perr}\includegraphics[width=.45\textwidth]{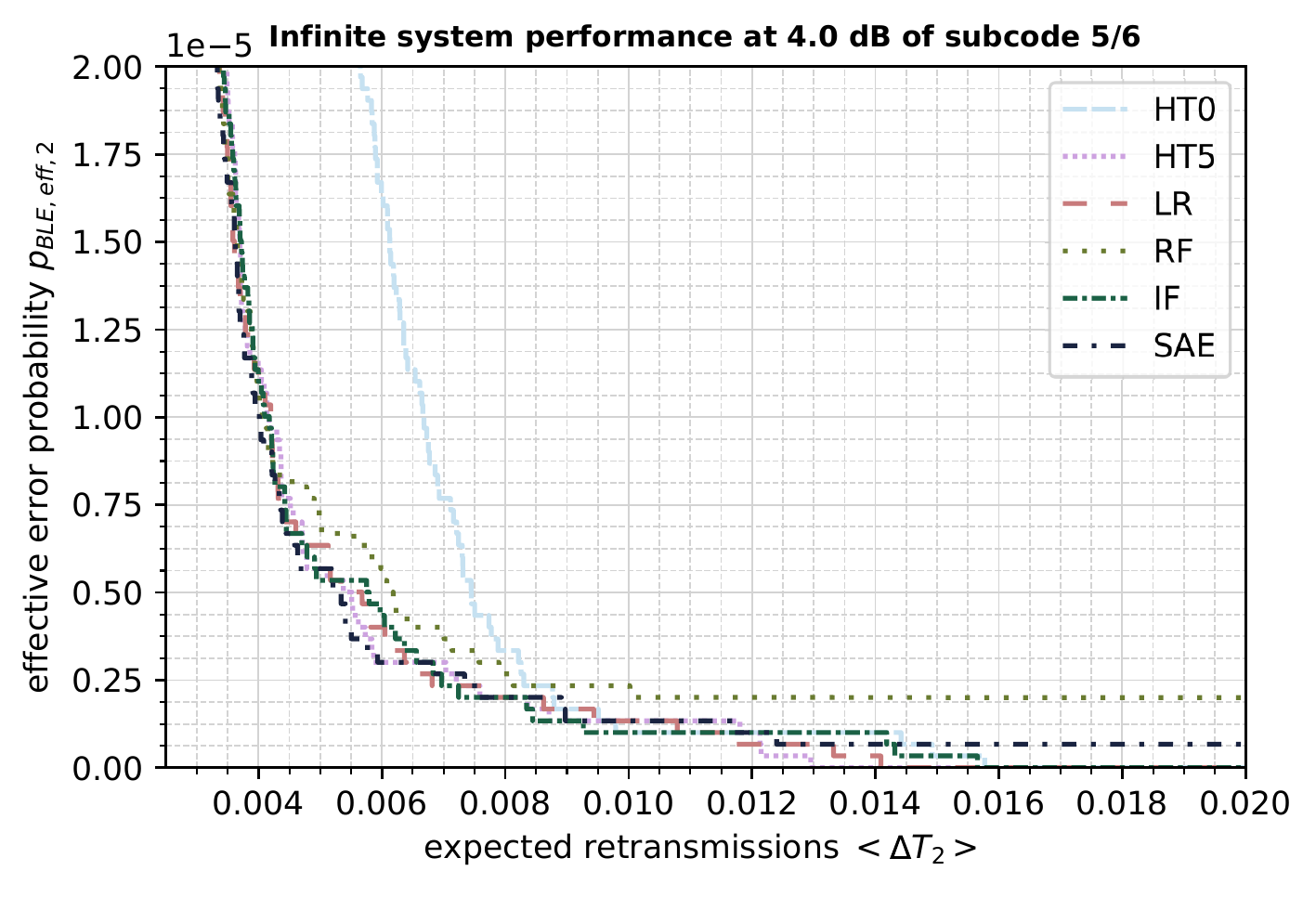}}
  \subfigure[3.5 dB subcode 1/2]{\label{fig:vnr35_12_perr}\includegraphics[width=.45\textwidth]{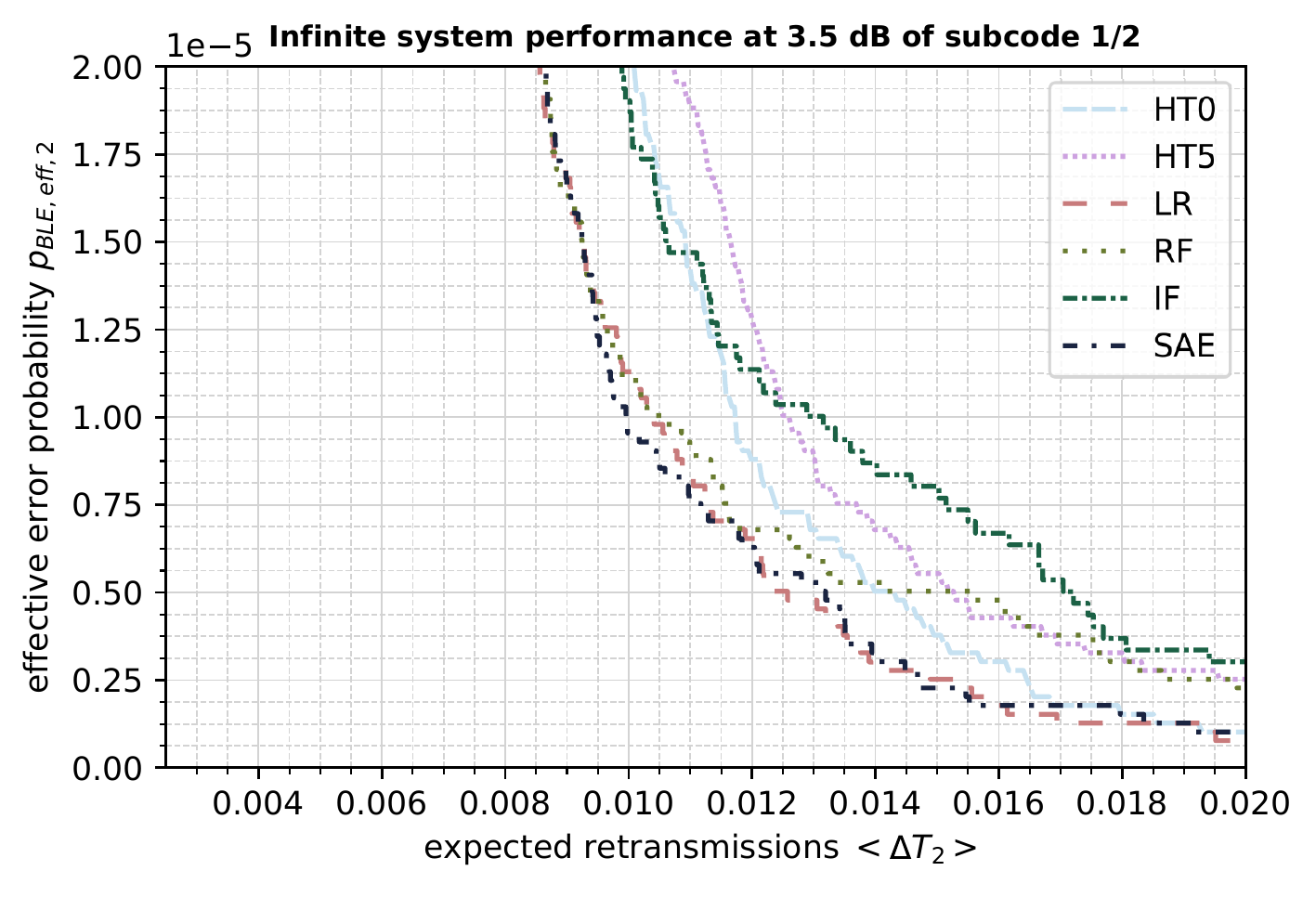}}
  \caption{Selected examples for system performance in the pedestrian channel for two-retransmission E-HARQ with unlimited system resources.}
  \label{fig:basic_vnr_perr}
  \end{figure*}

Irrespective of SNR, subcode length and underlying pedestrian or vehicular channel model, we see an improvement in classification performance upon including history features with best results achieved by incorporating Euclidean distance features. History information  seems to lead to larger improvements in the pedestrian channel compared to the vehicular channel. This is in line with the channel conditions remaining unchanged for a longer time in the pedestrian compared to the vehicular case.

There are different caveats to this result. First of all, as discussed in Section~\ref{sec:features}, the Euclidean distance is only known to the receiver if the underlying codeword is known as it would be the case for a previous pilot transmission, which would however lead to latency overheads. Therefore the result including Euclidean history features most likely overestimates the improvements in classification performance that can be obtained from using history features. Secondly, the use of history features is at tension with the assumption of an independent channel realization for the retransmission in the sense of $P_\mathrm{e'|e}=P_\mathrm{e}$ as used in our system model. It is very unlikely that the improvements in classification performance can compensate the loss of approximately one order of magnitude in the error rate for the retransmission of $P_\mathrm{e'|e}\approx 10^{-2}$ using the same channel compared to the baseline BLER of the order of $10^{-3}$ for an independent retransmission. Therefore the system level analysis is carried out using VNR-features only. Nevertheless the results put forward here stress the prospects of further investigations of features that explicitly characterize the channel state such as explicit channel state information that could have been obtained by a pilot transmission preceding the transmission.

\subsection{System performance}
\label{sec:systemperformance}

We start by discussing system performance based on the simple probabilistic model for E-HARQ with unlimited system resources as introduced in Section~\ref{sec:probmodel}. The results are obtained straightforwardly from the FNR-FPR-curves presented in Section~\ref{sec:vnrfeatures} using (\ref{eq:peff}) and (\ref{eq:ETr}) or the corresponding generalizations for multiple retransmissions (\ref{eq:peffn}) and (\ref{eq:ETrn}). Here we adopt $P_{\mathrm{e' | e}}=P_\mathrm{e}$ as in Section~\ref{sec:systemmodel}. Here we present results for two retransmissions that are possible for E-HARQ in both TTI scenarios discussed in Section~\ref{sec:systemmodel}. In fact, increasing the number of retransmissions beyond two does not lead to further noticeable improvements in the given FNR range. In all cases effective BLERs of the order $10^{-5}$ are attainable. Decreasing the subcode length from 5/6 to 1/2 while keeping the same effective BLER of $1\cdot 10^{-5}$ as a definite example requires an increase of 40\% and 45\% in retransmissions at SNR \unit[4]{dB} and \unit[3]{dB} respectively. Correspondingly, decreasing the SNR for fixed subcode length from \unit[4]{dB} to \unit[3]{dB} while again keeping the effective BLER fixed leads to an overhead of 70\% and 77\% in retransmissions for subcode 5/6 and 1/2 respectively. However, as discussed in Section~\ref{sec:finitesystemsize}, the presented effective BLERs only represent theoretical lower bounds for actual packet failure rates that are achievable in actual systems as they do not incorporate scheduling effects. In this infinite system setting there is no distinguished working point for the classifier and the only way of discriminating between different classifiers in the system setting is to rank by the number of expected transmissions for fixed effective error probability. 
\begin{figure*}[ht]
\centering     
\subfigure[3.5 dB subcode 1/2 (high load, $P_\mathrm{A,UE} = 0.36$)]{\label{fig:vnr35_12_high_system_perf}\includegraphics[width=.45\textwidth]{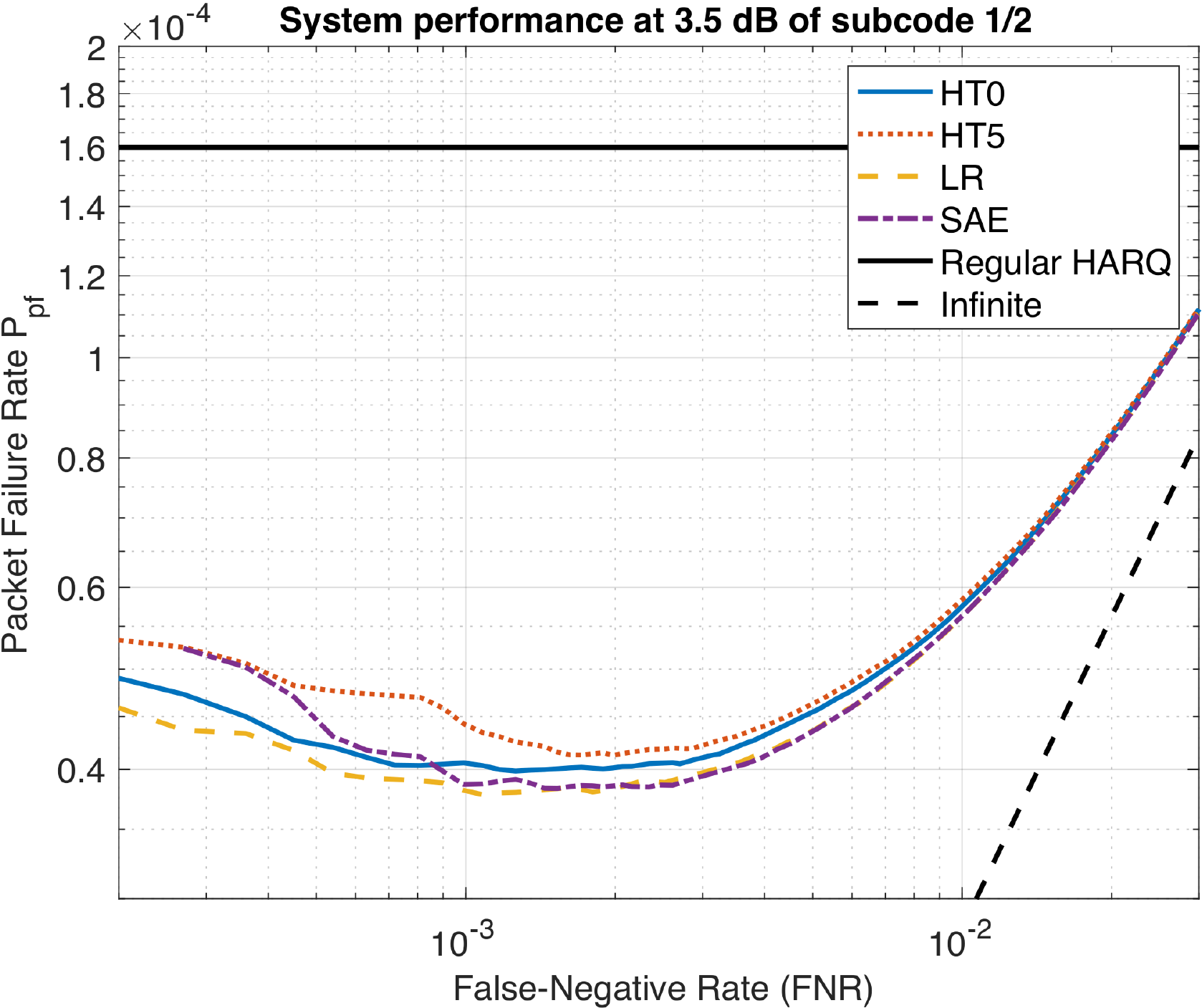}}
\subfigure[3.5 dB subcode 1/2 (medium load, $P_\mathrm{A,UE} = 0.30$)]{\label{fig:vnr35_12_low_system_perf}\includegraphics[width=.45\textwidth]{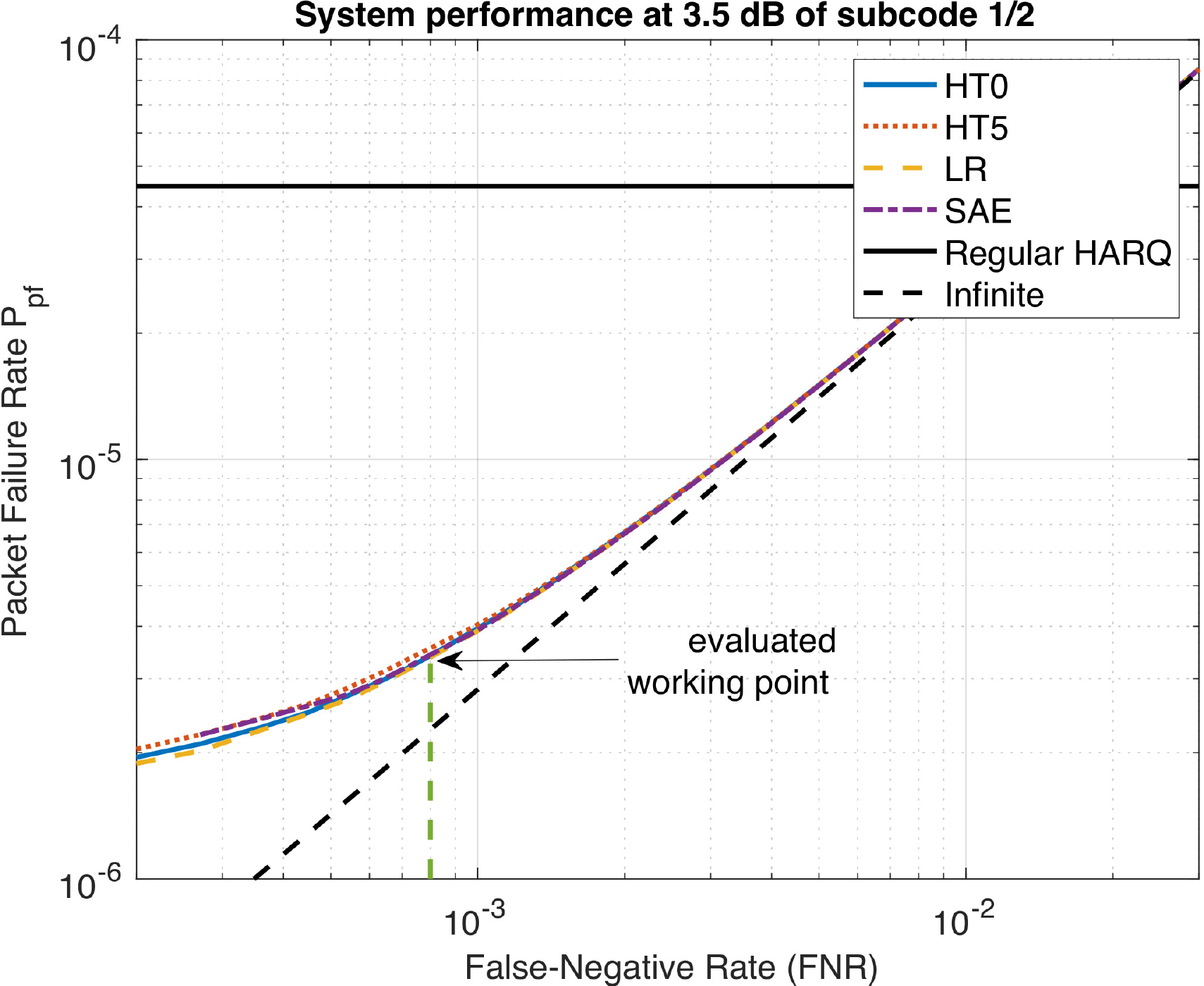}}
\subfigure[4.0 dB subcode 5/6 (high load, $P_\mathrm{A,UE} = 0.36$)]{\label{fig:vnr40_56_high_system_perf}\includegraphics[width=.45\textwidth]{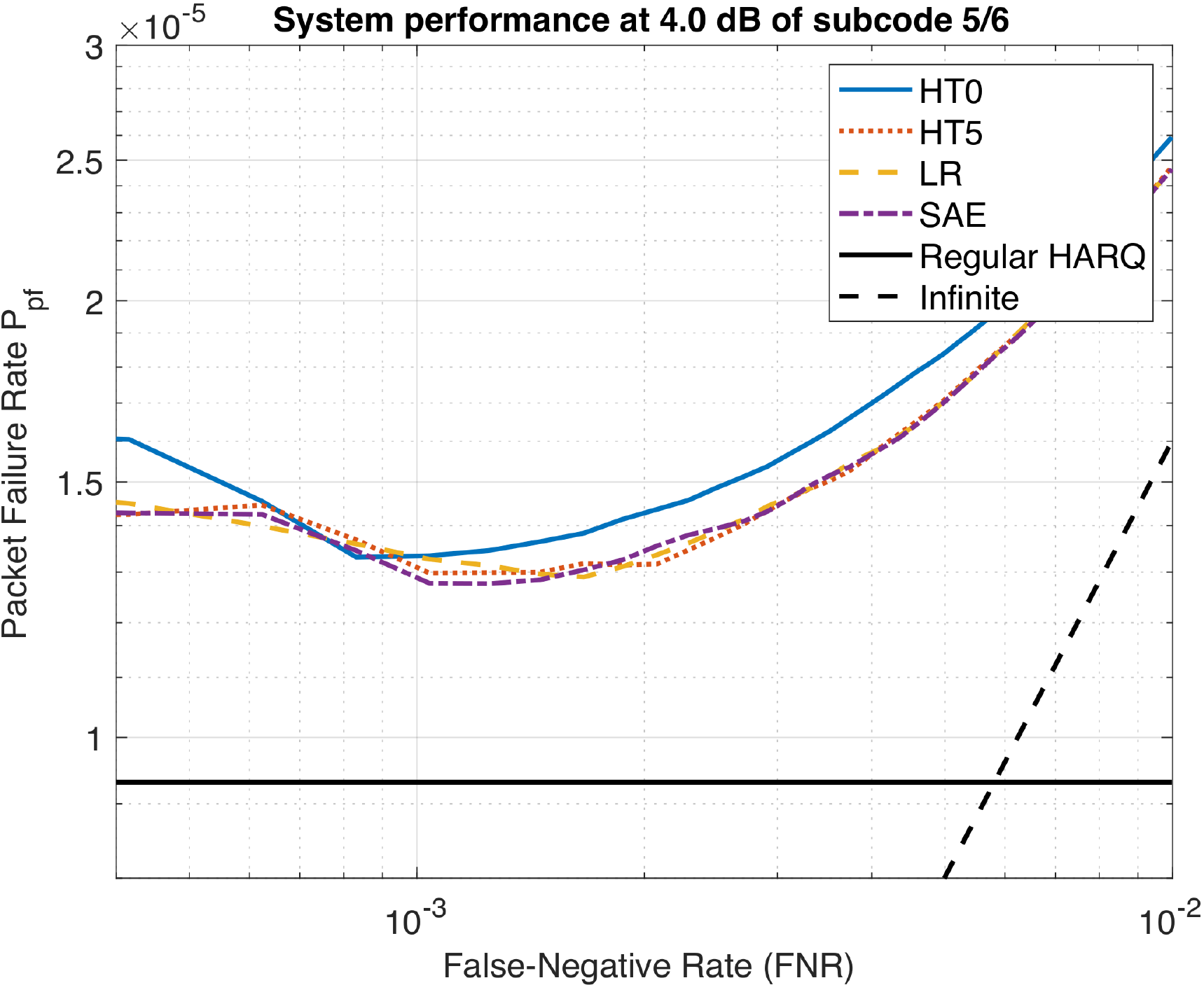}}
\subfigure[4.0 dB subcode 5/6 (medium load, $P_\mathrm{A,UE} = 0.30$)]{\label{fig:vnr40_56_low_system_perf}\includegraphics[width=.45\textwidth]{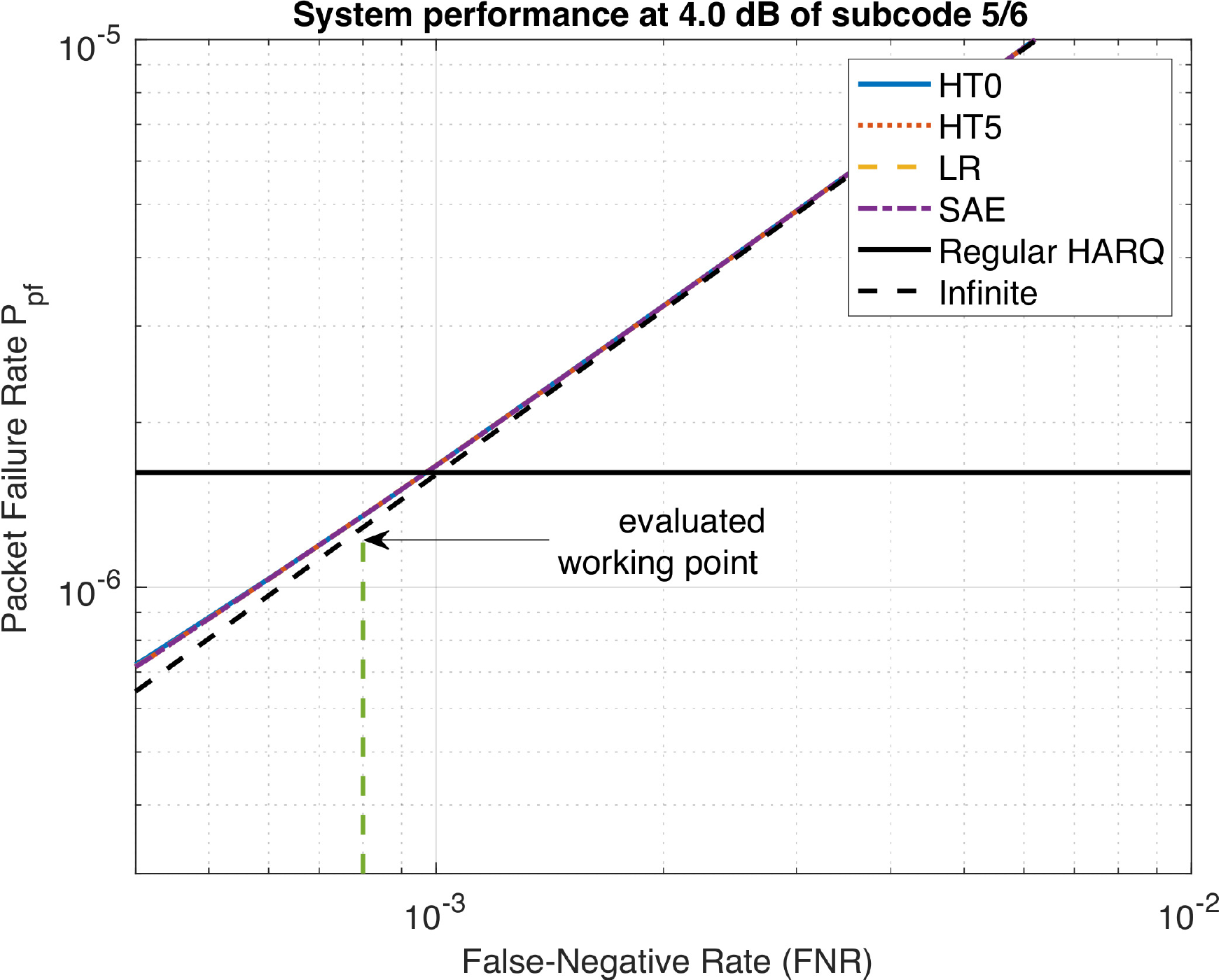}}
\caption{Exemplary system performance comparison for rate 1/2 and 5/6 prediction schemes in high load and medium load scenarios (green dashed line indicates $\mathrm{FNR}_\mathrm{eval}$).}
\label{fig:system_perf}
\end{figure*}

\begin{figure*}[ht]
  \centering     
  \subfigure[3.0 dB subcode 1/2 (high load, $P_\mathrm{A,UE} = 0.36$ and $T_c = 3$)]{\label{fig:vnr30_12_high_3_system_perf}\includegraphics[width=.45\textwidth]{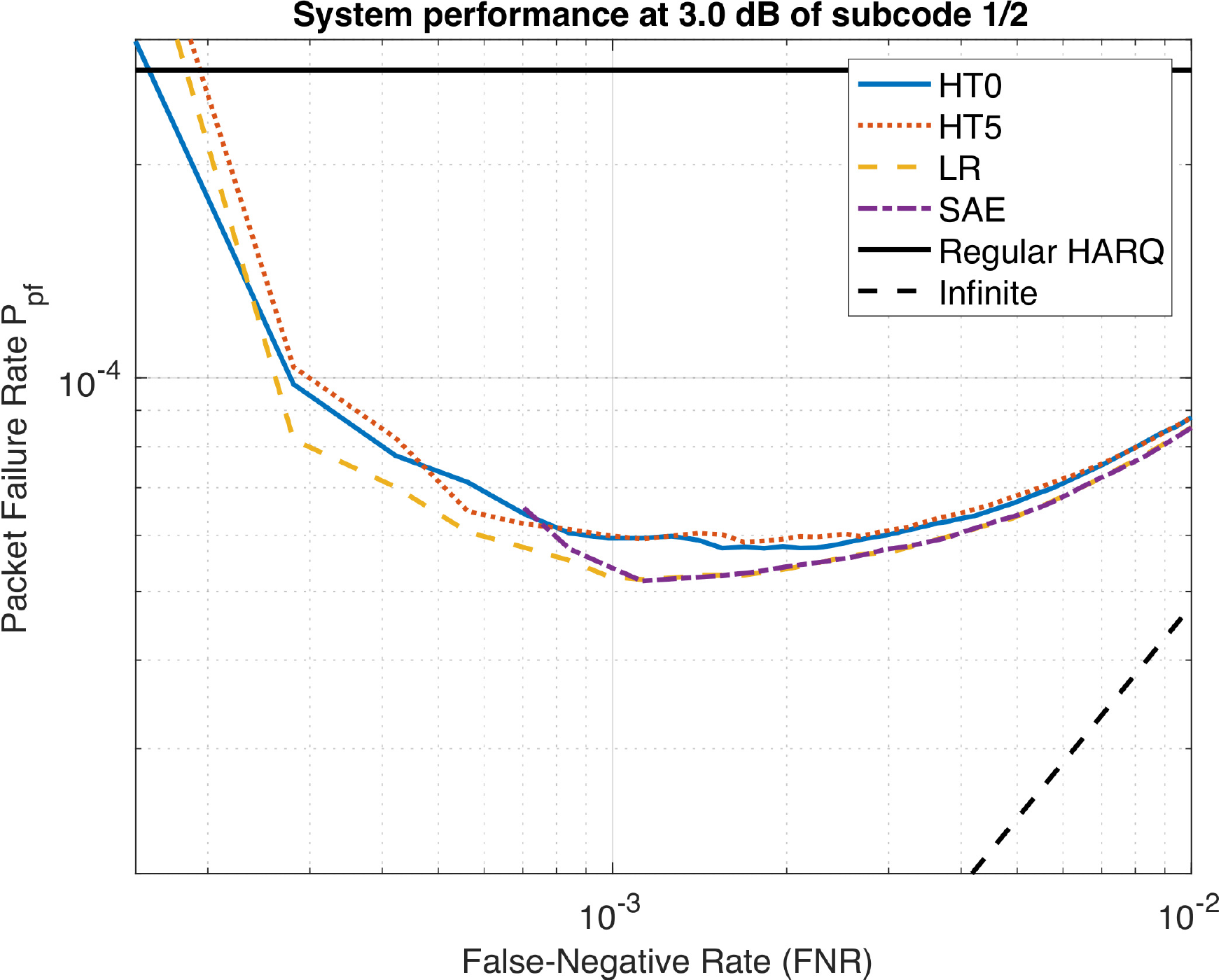}}
  \subfigure[3.0 dB subcode 1/2 (high load, $P_\mathrm{A,UE} = 0.36$ and $T_c = 4$)]{\label{fig:vnr30_12_high_4_system_perf}\includegraphics[width=.45\textwidth]{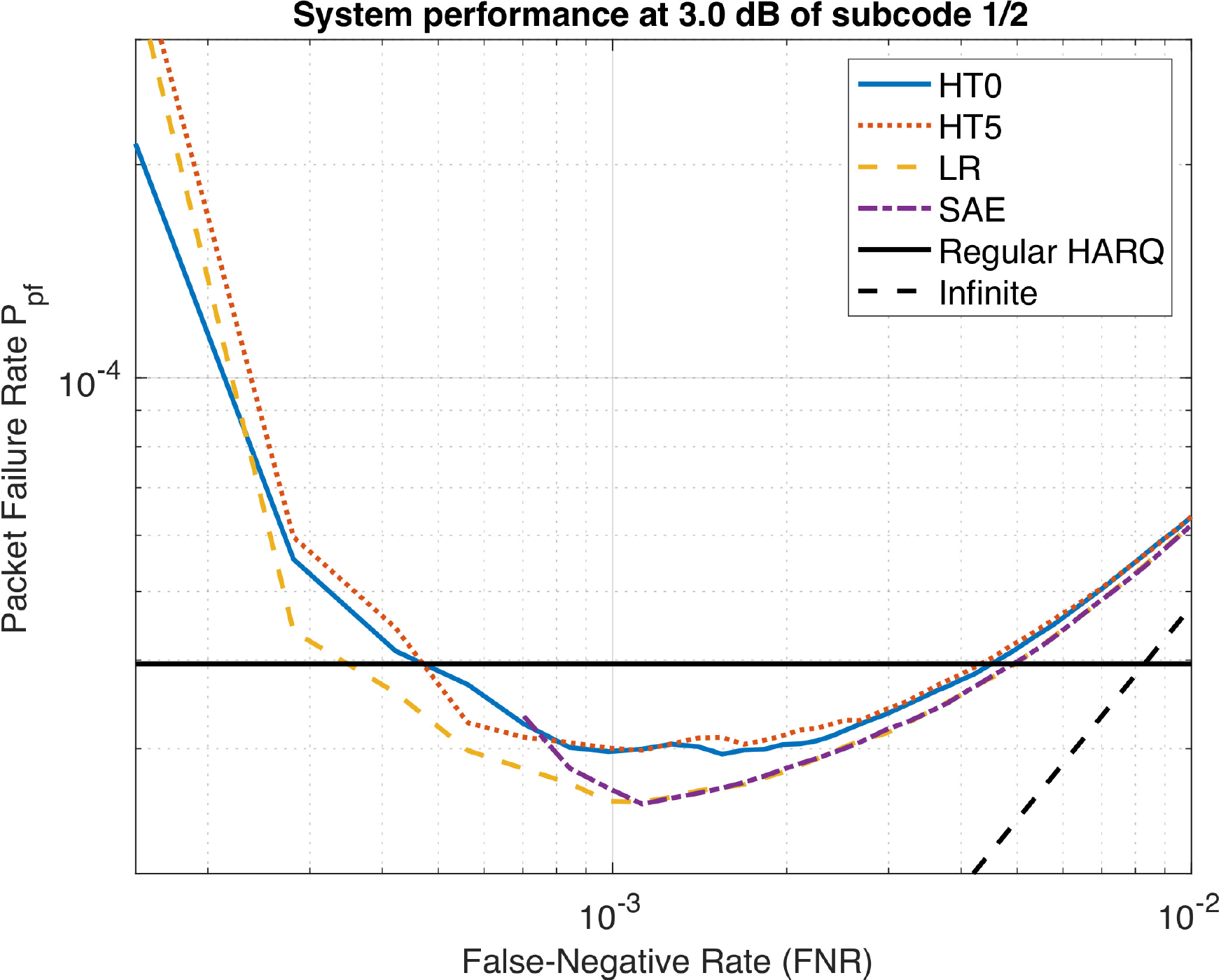}}
  \caption{Effects of the scheduling gain in the high load regime for the strict and relaxed latency constraint.}
  \label{fig:system_perf_resource_short}
  \end{figure*}
Figure~\ref{fig:system_perf} shows exemplary results of the packet failure rate over the \ac{FNR} of the \ac{E-HARQ} schemes under medium ($P_\mathrm{A,UE} = 0.3$) and high system load ($P_\mathrm{A,UE} = 0.36$) together with the regular \ac{HARQ}-baseline and the infinite system results from (\ref{eq:peffn}). The upper Figures~\ref{fig:vnr35_12_high_system_perf} and \ref{fig:vnr35_12_low_system_perf} show the long \ac{TTI} design, as described in Section~\ref{sec:systemmodel}, at \unit[3.5]{dB}. For the high load (Figure~\ref{fig:vnr35_12_high_system_perf}) as well as the medium load (Figure~\ref{fig:vnr35_12_low_system_perf}) scenarios, the \ac{E-HARQ} schemes achieve a superior performance compared to the regular \ac{HARQ} thanks to the additional retransmission which is possible within the same latency constraint. However, a packet failure rate less than $10^{-5}$ is only achieved in the medium load scenario. Here, we note that the actual performance of the \ac{E-HARQ} schemes is approximated well by the approach with infinite resources, at least for high packet failure rates above $10^{-5}$. Only in the lower region an attenuation of the decrease is visible, whereas all prediction schemes achieve a comparable performance. In the high load scenario in Figure~\ref{fig:vnr35_12_high_system_perf}, we see the trade-off behavior, discussed in Section~\ref{sec:finitesystemsize}. The packet failure rate decreases only up to a certain minimum at the optimal \ac{FNR}-\ac{FPR} trade-off and starts increasing after passing that point. So, lowering the \ac{FNR} further after passing that point increases the packet failure due to the resource shortage. In this region, the actual performance of the prediction schemes becomes critical. Hence, SAE and LR have the lowest optimum. HT0 and HT5 perform worse at their optimal operation points, whereas HT0 is still performing better than HT5.

The resource shortage effect is clearly visible in Figure~\ref{fig:system_perf_resource_short}, where the same load is applied in both scenarios but the latency constraint is relaxed in Figure~\ref{fig:vnr30_12_high_4_system_perf}. As obvious in Figure~\ref{fig:vnr30_12_high_3_system_perf}, the packet failure rate for all schemes is far away from the targeted packet failure rate of $10^{-5}$. With a relaxed latency constraint, as shown in Figure~\ref{fig:vnr30_12_high_4_system_perf}, the performance is closer to the target packet failure rate. This improvement is explainable by two effects. First, the E-HARQ schemes benefit from the additional retransmission, which is possible in the relaxed latency constraint and thus in total achieve still a better performance than the regular HARQ. However, the gap is smaller compared to the normal latency constraint. Especially in the high load scenario, the regular HARQ profits from the increased scheduling flexibility although it can only perform the same number of HARQ retransmissions.
The resource shortage effect is also observable for the regular \ac{HARQ} performance comparing the medium load and the high load scenarios. It is notable that the regular \ac{HARQ} could at least achieve a packet failure rate less than $10^{-4}$ in the medium load scenario, whereas it is performing even worse in the high load scenario. We can see that even more clearly in the short \ac{TTI} design in Figures~\ref{fig:vnr40_56_high_system_perf} and \ref{fig:vnr40_56_low_system_perf}. In the medium load scenario in Figure~\ref{fig:vnr40_56_low_system_perf}, the regular HARQ achieves a packet failure rate of almost $10^{-6}$, which corresponds approximately to the ideal performance of HARQ. In this system setup the regular HARQ makes use of the whole scheduling flexibility and thus, at least for the medium load scenario, the influence of scheduling probabilities can be neglected for the regular HARQ. Despite the limited scheduling flexibilities of the E-HARQ schemes, they achieve a better performance than the regular HARQ. However, this changes in the high load scenario in Figure~\ref{fig:vnr40_56_high_system_perf}. Here, we observe that the regular HARQ benefits from its scheduling gain and thus, achieves the lower packet failure rate. In the high load scenario, we see that all prediction schemes achieve a similar performance, except the HT0 which is remarkably less performing than the others. 

As already visible in the previous results, there is no clear winning scheme for all the scenarios. However, to compare the overall performance of the schemes, we  introduce the total score $t_s = \sum_{t} \log_{10}\frac{P_{\mathrm{pfr},s,t}}{\min_{s }P_{\mathrm{pfr},s,t}}$, where $t$ is the enumerator over all SNRs and prediction rates and $s$ is the enumerator over all HARQ schemes. In Table~\ref{tab:system_perf} we present the results for all scenarios, where the ''$<$'' sign indicates that an \ac{FNR} larger than the optimal \ac{FNR} has been used for evaluations. As already notable in Figure~\ref{fig:system_perf}, the available data does not allow arbitrary small \acp{FNR} and thus the optimal operation point cannot be reached for the medium load case. Hence, we used $\mathrm{FNR}_\mathrm{eval}= 8 \cdot 10^{-4}$ for the medium load evaluations since it provides a sufficiently reliable estimation. The evaluation at fixed FNR underestimates the overall performance compared to regular HARQ but allow a reliable ranking between different classifiers. Obviously, for reaching the optimal point of operation more data is required in the medium load case.

Nevertheless, in the medium load regime, LR achieves by far the best overall performance. The other \ac{E-HARQ} schemes achieve a similar performance, where HT0 is able to achieve a slightly better performance than the other two. Interestingly here, SAE has a worse performance compared to LR although it was the best performing classifier in the previous section. A closer inspection reveals that for very low \ac{FNR} SAE cannot keep up with the other classifiers. Especially that region, being not relevant for the performance metrics of the previous section, explains the contradicting results. However, the expected performance for SAE is observed going to the high load regime. Here, SAE and LR are the best performing \ac{E-HARQ} schemes far ahead HT0, HT5 and regular \ac{HARQ}. As already noted in Figure~\ref{fig:system_perf}, in the high load regime the performance at higher \acp{FNR} is key. Hence, SAE is again in a well-operating region. In this region, we also note that HT0 is performing the worst among the classifiers though having the second-best performance in the medium load regime.

Summa summarum, E-HARQ is able to achieve large gains in means of packet failure rate compared to regular HARQ under latency constraints. The rather small performance difference in medium and high load scenarios stems from the resolution issue at very small \acp{FNR} in the medium load scenarios. Especially, LR is a promising approach, which achieves a good overall performance in high load as well as medium load regimes. The SAE as best-performing algorithm in the high-load case and the more extendable approach compared to LR might provide a viable alternative if the performance at very low \acp{FNR} is improved.

\begin{table*}[ht]
\centering
\caption{Comparing system performance ($P_\text{pf}$) at their optimal FNR-FPR trade-off, as described in Section~\ref{sec:finitesystemsize}.}
\begin{tabular}{c|c||c|c|c|c|c}
& scenario & regular HARQ & HT0 & HT5 & LR & SAE \\
\hline
\hline                                                                                                                 
 \parbox[t]{2mm}{\multirow{8}{*}{\rotatebox[origin=c]{90}{medium load}}}                                                                     
&3.0dB 1/2 ped & $8.59 \cdot 10^{-5}$ & $<6.26 \cdot 10^{-6}$ & $<6.28 \cdot 10^{-6}$ & $<\mathbf{6.13 \cdot 10^{-6}}$ & $<6.18 \cdot 10^{-6}$ \\
\cline{2-7}                                                                                                                                  
&3.5dB 1/2 ped & $4.48 \cdot 10^{-5}$ & $<3.43 \cdot 10^{-6}$ & $<3.80 \cdot 10^{-6}$ & $<\mathbf{3.40 \cdot 10^{-6}}$ & $<3.67 \cdot 10^{-6}$ \\
\cline{2-7}                                                                                                                                  
&4.0dB 1/2 ped & $2.36 \cdot 10^{-5}$ & $<1.93 \cdot 10^{-6}$ & $<1.92 \cdot 10^{-6}$ & $<\mathbf{1.88 \cdot 10^{-6}}$ & $<1.92 \cdot 10^{-6}$ \\
\cline{2-7}                                                                                                                                  
&3.0dB 5/6 ped & $8.59 \cdot 10^{-5}$ & $<6.36 \cdot 10^{-6}$ & $<6.15 \cdot 10^{-6}$ & $<6.08 \cdot 10^{-6}$ & $<\mathbf{5.99 \cdot 10^{-6}}$ \\
\cline{2-7}                                                                                                                                  
&3.5dB 5/6 ped & $7.86 \cdot 10^{-6}$ & $<2.25 \cdot 10^{-6}$ & $<\mathbf{2.22 \cdot 10^{-6}}$ & $<\mathbf{2.22 \cdot 10^{-6}}$ & $<2.23 \cdot 10^{-6}$ \\
\cline{2-7}                                                                                                                                  
&4.0dB 5/6 ped & $1.62 \cdot 10^{-6}$ & $<\mathbf{1.40 \cdot 10^{-6}}$ & $<\mathbf{1.40 \cdot 10^{-6}}$ & $<\mathbf{1.40 \cdot 10^{-6}}$ & $<\mathbf{1.40 \cdot 10^{-6}}$ \\
\cline{2-7}                                                                                                                                  
&3.5dB 1/2 veh & $4.48 \cdot 10^{-5}$ & $<3.47 \cdot 10^{-6}$ & $<3.48 \cdot 10^{-6}$ & $<\mathbf{3.37 \cdot 10^{-6}}$ & $<3.68 \cdot 10^{-6}$ \\
\hhline{~======}                                                                                                                                    
&total score $t_s$ & 6.2577 & 0.0685 & 0.0936 & \textbf{0.0075} & 0.0866 \\    
\hline                                                                                                 
\hline                                                                                                                 
\parbox[t]{2mm}{\multirow{8}{*}{\rotatebox[origin=c]{90}{high load}}}
&3.0dB 1/2 ped & $2.72 \cdot 10^{-4}$ & $5.75 \cdot 10^{-5}$ & $5.87 \cdot 10^{-5}$ & $5.20 \cdot 10^{-5}$ & $\mathbf{5.17 \cdot 10^{-5}}$ \\
\cline{2-7}                                                                                                                                  
&3.5dB 1/2 ped & $1.60 \cdot 10^{-4}$ & $3.99 \cdot 10^{-5}$ & $4.13 \cdot 10^{-5}$ & $\mathbf{3.78 \cdot 10^{-5}}$ & $3.83 \cdot 10^{-5}$ \\
\cline{2-7}                                                                                                                                  
&4.0dB 1/2 ped & $9.56 \cdot 10^{-5}$ & $2.94 \cdot 10^{-5}$ & $2.88 \cdot 10^{-5}$ & $\mathbf{2.76 \cdot 10^{-5}}$ & $2.81 \cdot 10^{-5}$ \\
\cline{2-7}                                                                                                                                  
&3.0dB 5/6 ped & $2.72 \cdot 10^{-4}$ & $5.59 \cdot 10^{-5}$ & $4.99 \cdot 10^{-5}$ & $4.89 \cdot 10^{-5}$ & $\mathbf{4.70 \cdot 10^{-5}}$ \\
\cline{2-7}                                                                                                                                  
&3.5dB 5/6 ped & $1.61 \cdot 10^{-5}$ & $2.05 \cdot 10^{-5}$ & $\mathbf{1.61 \cdot 10^{-5}}$ & $1.65 \cdot 10^{-5}$ & $1.68 \cdot 10^{-5}$ \\
\cline{2-7}                                                                                                                                  
&4.0dB 5/6 ped & $\mathbf{9.32 \cdot 10^{-6}}$ & $1.33 \cdot 10^{-5}$ & $1.30 \cdot 10^{-5}$ & $1.29 \cdot 10^{-5}$ & $1.28 \cdot 10^{-5}$ \\
\cline{2-7}                                                                                                                                  
&3.5dB 1/2 veh & $1.60 \cdot 10^{-4}$ & $3.88 \cdot 10^{-5}$ & $4.06 \cdot 10^{-5}$ & $\mathbf{3.64 \cdot 10^{-5}}$ & $\mathbf{3.64 \cdot 10^{-5}}$ \\
\hhline{~======}                                                                                                                               
& \textbf{total score $t_s$} & 3.2918 & 0.4599 & 0.3306 & 0.1713 & \textbf{0.1703} \\   
\end{tabular}
\label{tab:system_perf}
\end{table*}

\section{Summary and Conclusions}
\label{sec:summary}
In this work we investigated machine learning techniques for \ac{E-HARQ} by means of more 
elaborate classification methods to predict the decoding result ahead of the final decoder iteration.
We demonstrated that more complex estimators such as logistic regression or supervised autoencoder that exploit the evolution of the subcodeword during the first
few decoder iterations lead to quantitative improvements in the prediction performance over baseline results across different SNR and channel conditions. 
We put forward a simple probabilistic model and a more elaborate system model incorporating scheduling effects to evaluate system performance in a realistic environment.
In this way we were able to demonstrate the practical feasibility of reaching effective packet error rates of the order $10^{-5}$ as required for URLLC across a range of different SNRs, subcode lengths and system loads.
More importantly, we showed that enabling more \ac{HARQ} iterations by introducing E-HARQ improves the overall reliability over regular \ac{HARQ} under strict maximum latency constraints. We hope to extend this analysis to other URLLC-relevant scenarios such as MIMO in the future.

Further improvements of the classification performance are conceivable extending the approach presented in this work. Our results suggest 
that history features incorporating channel information from previous transmissions positively influence the classification performance
but remain to be investigated in more detail. Similarly it seems very likely that classification algorithms could profit from intra-message features that go 
beyond the simple averaging features such as VNRs considered in this work, which ideally directly incorporate the code structure of the 
underlying channel code. However, such features suffer from high dimensionality and large correlations. Here a challenge remains to identify the most discriminative set of input features and appropriate 
classification algorithms to further improve the classification performance.

Ultimately, more advanced classification algorithms, which are within reach using techniques presented in this work, might 
allow more fine-grained feedback instead of a binary NACK/ACK response. Incorporating this information on the level of the feedback protocol
would allow to design custom feedback schemes with potentially large latency gains.
\bibliographystyle{IEEEtran}
\bibliography{lib}

\appendices

\section{Probabilistic model for multiple-retransmission \ac{E-HARQ}}
\label{app:eharqmult}
In this section, we present the generalization of the results from Section~\ref{sec:probmodel}. These are obtained straightforwardly using the same formalism as above.
The generalization of the effective error probability from (\ref{eq:peff}) to the case of $n$ retransmissions is given by the iterative relation
\begin{equation}
\label{eq:peffn}
p_{\mathrm{BLE,eff,n}} = P_\mathrm{e} P_\mathrm{H,e,1}\,,
\end{equation}
where we defined for $j \leq n$
\ifone
\begin{align}
P_\mathrm{H,e,j} =&
P_\mathrm{fn} + (1- P_\mathrm{fn}) P_{e^{(j)}|e^{(j-1)}\wedge\ldots\wedge e^{(0)}}\cdot P_\mathrm{H,e,j+1}\,,
\end{align} 
\else
\begin{align}
P_\mathrm{H,e,j} =&
P_\mathrm{fn} + (1- P_\mathrm{fn})\nonumber\\
&\cdot P_{e^{(j)}|e^{(j-1)}\wedge\ldots\wedge e^{(0)}}\cdot P_\mathrm{H,e,j+1}\,,
\end{align} 
\fi
and otherwise $P_\mathrm{H,e,j}=1$,
which reduces to (\ref{eq:peff}) for $n=1$. For simplicity we can work with independent retransmissions i.e.\ $P_{e^{(j)}|e^{(j-1)}\wedge\ldots\wedge e^{(0)}}=P_\mathrm{e}$, where we used the shorthand notation $P_{e^{(j)}|e^{(j-1)}\wedge\ldots\wedge e^{(0)}}\equiv P_{(e^{(j)}=1)|(e^{(j-1)}=1)\wedge\ldots\wedge (e^{(0)}=1)} $. Explicit expression for up to three retransmissions are in this case given by
\ifone
\begin{align}
p_{\mathrm{BLE,eff,1}} =& P_\mathrm{e}\left(P_\mathrm{fn}+(1-P_\mathrm{fn})P_\mathrm{e}\right)\,,\\
p_{\mathrm{BLE,eff,2}} =& P_\mathrm{e}(P_\mathrm{fn}+(1-P_\mathrm{fn})P_\mathrm{e}\left(P_\mathrm{fn}+(1-P_\mathrm{fn})P_\mathrm{e}\right))\,,\\
p_{\mathrm{BLE,eff,3}} =& P_\mathrm{e}(P_\mathrm{fn}+(1-P_\mathrm{fn})P_\mathrm{e}(P_\mathrm{fn}+(1-P_\mathrm{fn})P_\mathrm{e}\left(P_\mathrm{fn}+(1-P_\mathrm{fn})P_\mathrm{e}\right)))\,.
\end{align}
\else
\begin{align}
p_{\mathrm{BLE,eff,1}} =& P_\mathrm{e}\left(P_\mathrm{fn}+(1-P_\mathrm{fn})P_\mathrm{e}\right)\,,\\
p_{\mathrm{BLE,eff,2}} =& P_\mathrm{e}(P_\mathrm{fn}+(1-P_\mathrm{fn})P_\mathrm{e}\nonumber\\
&\quad\cdot\left(P_\mathrm{fn}+(1-P_\mathrm{fn})P_\mathrm{e}\right))\,,\\
p_{\mathrm{BLE,eff,3}} =& P_\mathrm{e}(P_\mathrm{fn}+(1-P_\mathrm{fn})P_\mathrm{e}\nonumber\\
&\quad\cdot(P_\mathrm{fn}+(1-P_\mathrm{fn})P_\mathrm{e}\nonumber\\
&\quad\cdot\left(P_\mathrm{fn}+(1-P_\mathrm{fn})P_\mathrm{e}\right)))\,.
\end{align}
\fi
If we denote the set of binary sequences of length $n$ by $\mathcal{S}_n$, the probability $P_\mathrm{r,n}$ for having $n$ retransmissions is given by
\ifone
\begin{align}
\label{eq:prn}
P_\mathrm{r,n}=&\sum_{(x_0, x_1,\ldots x_{n-1})\in \mathcal{S}_{n}} \prod_{i=0}^{n-1} (1-P_\mathrm{fn})^{x_i} {P_\mathrm{fp}}^{1-x_i}\prod_{j=0}^{n-1} P_{(e^{(j)}=x_j)|(e^{(j-1)}=x_{j-1})\wedge\ldots\wedge(e^{(0)}=x_0)}\,,
\end{align}
\else
\begin{align}
\label{eq:prn}
P_\mathrm{r,n}=&\sum_{(x_0, x_1,\ldots x_{n-1})\in \mathcal{S}_{n}} \prod_{i=0}^{n-1} (1-P_\mathrm{fn})^{x_i} {P_\mathrm{fp}}^{1-x_i}\nonumber\\
&\prod_{j=0}^{n-1} P_{(e^{(j)}=x_j)|(e^{(j-1)}=x_{j-1})\wedge\ldots\wedge(e^{(0)}=x_0)}\,,
\end{align}
\fi
which again reduces to (\ref{eq:pr1}) for $n=1$. Again we may set $P_{(e^{(j)}=x_j)|(e^{(j-1)}=x_{j-1})\wedge\ldots\wedge(e^{(0)}=x_0)}=P_\mathrm{e}$ for independent transmissions. In this case Eq.~\ref{eq:prn} simplifies to
\begin{equation}
P_\mathrm{r,n}= \left(P_\mathrm{e}(1-P_\mathrm{fn})+(1-P_\mathrm{e})P_\mathrm{fp}\right)^n\, .
\end{equation}
The total number of expected transmissions $\langle \Delta T_n \rangle$ is then simply given by
\begin{equation}
\label{eq:ETrn}
\langle \Delta T_n \rangle = \sum_{i=1}^n i\cdot P_\mathrm{r,i}\,.
\end{equation}
\section{Supervised autoencoder for supervised anomaly detection}
\label{sec:SAE}
The supervised autoencoder is a neural-network-based supervised anomaly detection algorithm. It enjoys a number of advantages compared to for example shallow neural network classifiers applied directly to the input data that arise from the fact that the classifier is not applied to the data directly but rather to the bottleneck features of an autoencoder. Therefore it is able to work in heavily imbalanced scenarios as the one considered in this work and does not suffer from highly correlated input.

For the construction of the SAE we leverage the approach put forward in \cite{zong2018deep} albeit in a supervised anomaly detection setting. Similar to their work we use a regular multi-layer fully-connected autoencoder with $L_2$ loss as a backbone. 
  In addition, we jointly train a fully-connected classifier operating on the bottleneck features that is trained using cross entropy loss, see Figure~\ref{fig:autoencoder} The idea behind the joint training is to allows the autoencoder to not only build a reduced representation but also to build bottleneck features that contain most discriminative information for the classification task. We also experimented with using features derived from the reconstruction error (measured using cosine distance and reduced Euclidean distance) as additional input to the classifier as proposed in \cite{zong2018deep} but found no improvement.

\begin{figure}[ht]
  \centering
  \ifone
  \includegraphics[width=.5\columnwidth]{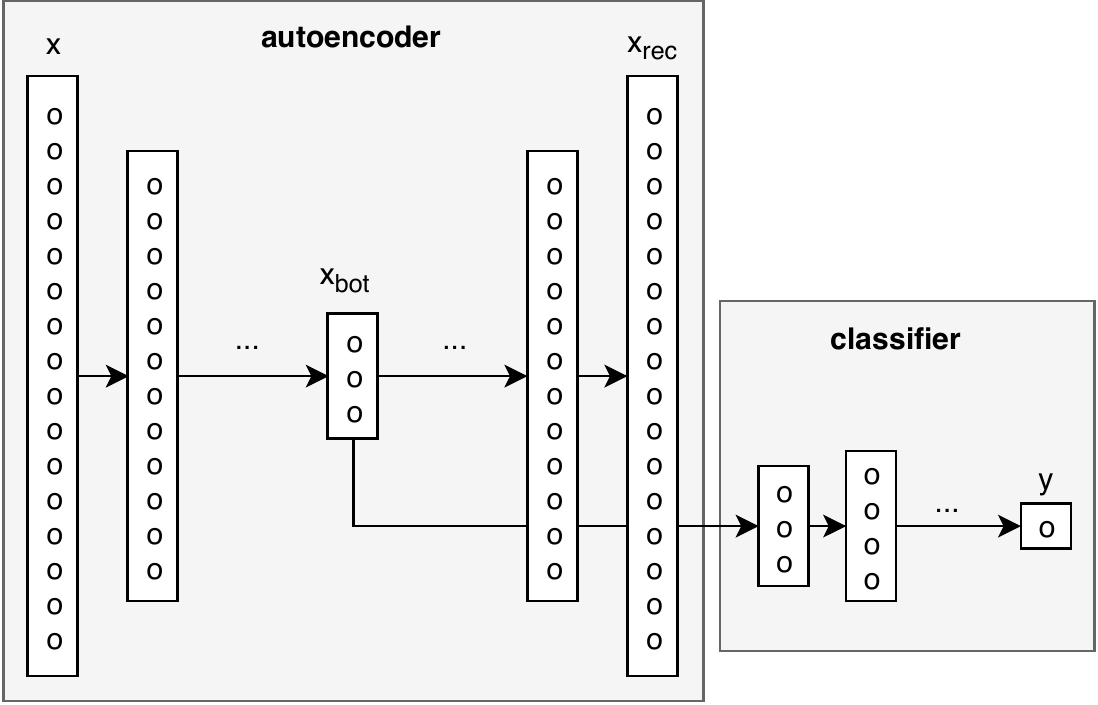}
  \else
  \includegraphics[width=.9\columnwidth]{autoencoder}
  \fi
  \caption{Architecture for supervised anomaly detection using a jointly trained supervised autoencoder ($x$: input, $x_\mathrm{rec}$: reconstructed input, $x_\mathrm{bot}$: bottleneck features, $y$: predicted label).}
  \label{fig:autoencoder}
\end{figure}

There are multiple ways of preventing overfitting in this setting: early stopping, reducing the bottleneck dimension, implementing the SAE as a denoising autoencoder \cite{vincent2008extracting} or regularization using dropout \cite{srivastava2014dropout}. In our case dropout regularization both in the classifier as well as in the autoencoder itself proved most effective.

The network configuration reads for the autoencoder  [FC($d$,25), FC(25,10), FC(10,3), FC(3,10), FC(10,25), Lin(25,$d$)] and for the classifier [FC(3,10), FC(10,5), Lin(5,2), SM] with FC(x,y) $\equiv$ [Lin(x,y), BN, ReLU, DO] and input dimension $d$. Here Lin(x,y) denotes a linear transformation layer, BN a Batch Normalization-layer \cite{ioffe2015batch}, ReLU a ReLU activation layer, DO a dropout layer at a dropout rate fixed via hyperparameter tuning (both 0.2) and SM a softmax activation layer. Optimization is performed using the Adam optimizer \cite{Kingma2014AdamAM} at learning rate 0.001. To stabilize training oversampling the minority class samples by a factor of 100 turned out to be beneficial.

\section{Scheduling probability of a system with finite resources}
\label{sec:sched}
In (\ref{eq:pfr}), $P(T_{\mathrm{j}} \leq T_\mathrm{c})$ highly depends on the load of the system, since it is mainly a scheduling problem. Based on the resource distribution $P_{\mathrm{res}}$ which is discussed in Appendix~\ref{sec:res_distr}, we can formulate the probability $P(T_{\mathrm{j}} \leq T_\mathrm{c})$ of scheduling the initial transmission arriving at time slot $t_0 > 0$ and $j-1$ retransmissions within a time constraint $T_c$ as follows,
\ifone
\begin{align}
P(T_{\mathrm{j}} \leq T_\mathrm{c}) &= \sum_{k_0=0}^{T_c - j T_\mathrm{RTT}-1} P_1(t_0, k_0) \sum_{k_1=k_0+T_\mathrm{RTT}}^{T_c - (j-1)T_\mathrm{RTT}-1}P_1(t_0+k_{0}+T_\mathrm{RTT},k_1-k_0-T_\mathrm{RTT})\cdots\nonumber\\
&\quad\quad\sum_{k_j=k_{j-1}+T_\mathrm{RTT}}^{T_c-1}P_1(t_0+k_{j-1}+T_\mathrm{RTT},k_j-k_{j-1}-T_\mathrm{RTT}) \,,
\label{eq:PTj}
\end{align}
\else
\begin{align}
P&(T_{\mathrm{j}} \leq T_\mathrm{c}) = \sum_{k_0=0}^{T_c - j T_\mathrm{RTT}-1} P_1(t_0, k_0) \nonumber\\
&\sum_{k_1=k_0+T_\mathrm{RTT}}^{T_c - (j-1)T_\mathrm{RTT}-1}P_1(t_0+k_{0}+T_\mathrm{RTT},k_1-k_0-T_\mathrm{RTT})\cdots\nonumber\\
&\sum_{k_j=k_{j-1}+T_\mathrm{RTT}}^{T_c-1}P_1(t_0+k_{j-1}+T_\mathrm{RTT},k_j-k_{j-1}-T_\mathrm{RTT})
\label{eq:PTj}
\end{align}
\fi
where $P_1(t_0, \Delta_t)$ is the probability that a packet that has arrived at $t_0$ is scheduled in time slot $t_0 + \Delta_t$. Under the assumption that the resource distribution function is not diverging, the initial argument of $P_1$ in (\ref{eq:PTj}) is set to $t_0$. As mentioned before, $P_1$ is the scheduling probability for an additional transmission assuming that this single transmission does not affect the system probabilities. So, this means that from the slots $t_0$ till the slot $t_0 + \Delta_t -1$ the system is fully loaded and the observed transmission is not scheduled (random scheduling). We allow only in slot $t_0 + \Delta_t$ a lower load or the random scheduler picks the observed transmission. Hence, this is expressed by,
\ifone
\begin{align}
P_1(t_0, \Delta_t)  &\approx \sum_{k_0 = N_{\mathrm{res}}}^{N_\mathrm{max}} P_{\mathrm{res}}(k_0, t_0)\left(1 - \frac{N_{\mathrm{res}}}{k_0+1}\right) \left( \sum_{k_1 = N_{\mathrm{res}}}^{N_\mathrm{max}} P_{\mathrm{res}}(k_1 | N_\mathrm{res})\left(1 - \frac{N_{\mathrm{res}}}{k_1+1}\right)\right)^{(\Delta_t - 1)}\nonumber\\
&\quad\quad\left( \sum_{k=0}^{N_{\mathrm{res}}-1} P_{\mathrm{res}}(k|N_\mathrm{res}) + \sum_{k = N_{\mathrm{res}}}^{N_\mathrm{max}} P_{\mathrm{res}}(k |N_\mathrm{res})\frac{N_{\mathrm{res}}}{k+1} \right)\,,
\end{align}
\else
\begin{align}
P_1&(t_0, \Delta_t)  \approx \sum_{k_0 = N_{\mathrm{res}}}^{N_\mathrm{max}} P_{\mathrm{res}}(k_0, t_0)\left(1 - \frac{N_{\mathrm{res}}}{k_0+1}\right) \nonumber\\
&\left( \sum_{k_1 = N_{\mathrm{res}}}^{N_\mathrm{max}} P_{\mathrm{res}}(k_1 | N_\mathrm{res})\left(1 - \frac{N_{\mathrm{res}}}{k_1+1}\right)\right)^{(\Delta_t - 1)}\nonumber\\
&\left( \sum_{k=0}^{N_{\mathrm{res}}-1} P_{\mathrm{res}}(k|N_\mathrm{res}) + \sum_{k = N_{\mathrm{res}}}^{N_\mathrm{max}} P_{\mathrm{res}}(k |N_\mathrm{res})\frac{N_{\mathrm{res}}}{k+1} \right)\,,
\end{align}
\fi
where $P_{\mathrm{res}}$ is the resource distribution function, which is discussed in more detail in Appendix~\ref{sec:res_distr}.
The scheduling probability $P_1$ is discussed in further detail in Appendix~\ref{sec:sched_prob}.

\section{Scheduling probability in a moderately loaded finite system}\label{sec:sched_prob}

The scheduling probability $P_1$ as the probability that a transmission arriving at $t_0$ is scheduled after $\Delta_t$ \acp{TTI} is given as
\ifone
\begin{align}
P_1(t_0, \Delta_t)  &= \sum_{k_0 = N_{\mathrm{res}}}^\infty P_{\mathrm{res}}(k_0, t_0)\left(1 - \frac{N_{\mathrm{res}}}{k_0+1}\right)\sum_{k_1 = N_{\mathrm{res}}}^\infty P_{\mathrm{res}}(k_1 | k_{0})\left(1 - \frac{N_{\mathrm{res}}}{k_1+1}\right) \cdots\nonumber\\
&\quad\quad\sum_{k_{(\Delta_t-1)} = N_{\mathrm{res}}}^\infty P_{\mathrm{res}}(k_{(\Delta_t-1)} | k_{(\Delta_t-2)})\left(1 - \frac{N_{\mathrm{res}}}{k_{(\Delta_t-1)}+1}\right)\nonumber\\
&\quad\quad\sum_{k=0}^{N_{\mathrm{res}}-1} P_{\mathrm{res}}(k | k_{(\Delta_t-1)}) + \sum_{k = N_{\mathrm{res}}}^\infty P_{\mathrm{res}}(k | k_{(\Delta_t-1)})\frac{N_{\mathrm{res}}}{k+1}\,.
\end{align}
\else
\begin{align}
P_1&(t_0, \Delta_t)  = \sum_{k_0 = N_{\mathrm{res}}}^\infty P_{\mathrm{res}}(k_0, t_0)\left(1 - \frac{N_{\mathrm{res}}}{k_0+1}\right)\nonumber\\
&\sum_{k_1 = N_{\mathrm{res}}}^\infty P_{\mathrm{res}}(k_1 | k_{0})\left(1 - \frac{N_{\mathrm{res}}}{k_1+1}\right) \cdots\nonumber\\
&\sum_{k_{(\Delta_t-1)} = N_{\mathrm{res}}}^\infty P_{\mathrm{res}}(k_{(\Delta_t-1)} | k_{(\Delta_t-2)})\left(1 - \frac{N_{\mathrm{res}}}{k_{(\Delta_t-1)}+1}\right)\nonumber\\
&\sum_{k=0}^{N_{\mathrm{res}}-1} P_{\mathrm{res}}(k | k_{(\Delta_t-1)}) + \sum_{k = N_{\mathrm{res}}}^\infty P_{\mathrm{res}}(k | k_{(\Delta_t-1)})\frac{N_{\mathrm{res}}}{k+1}\,.
\end{align}
\fi

As obvious, $P_1$ crucially depends on the resource distribution function  $P_{\mathrm{res}}(N, t)$, which is the probability that $N$ resources arrive at time slot $t$, and its probability distribution conditioned on the previous number of resource arrivals $P_{\mathrm{res}}(k_t | k_{t-1})$. The properties and formulation of this distribution is evaluated more in detail in Appendix~\ref{sec:res_distr}.

However, $P_1(t_0, \Delta_t)$ the exact formulation of $P_1(t_0, \Delta_t)$ poses computational problems due to the infinite sums and the exponential growth of computation for increasing $\Delta_t$. Hence, we introduce Lemma~\ref{lem:approx_res} to simplify the computation of the scheduling probability. 

\begin{lemma}\label{lem:approx_res}
For a moderately loaded system with $\sum_{k=0}^{N_\mathrm{max}}P(k,t) \approx 1$ and $N_\mathrm{max} \gtrapprox N_\mathrm{res}$, the resource distribution function is approximated for sufficiently large time slots $t$ by \begin{math}
P_{\mathrm{res}}(N_{t}, t) \approx \sum_{k=0}^{N_\mathrm{res}-1} P_\mathrm{res}(k, t-1) P_{\mathrm{res}}(N_{t} | k) + \sum_{k=N_\mathrm{res}}^{N_\mathrm{max}} P_\mathrm{res}(k, t-1) P_{\mathrm{res}}(N_{t} | N_\mathrm{res})\,.
\end{math}
\end{lemma}
\begin{proof}
Assuming a converging behavior of the resource distribution function, there exits a time slot $t_0$ and a lower bound $N_\mathrm{min}$ and an upper bound $N_\mathrm{max}$, such that $\sum_{k=N_\mathrm{min}}^{N_\mathrm{max}} P_{\mathrm{res}}(k, t) \approx 1$ for all $t \geq t_0$. Additionally for an non-heavily loaded system which is required for \ac{URLLC} traffic, we assume $N_\mathrm{max} \gtrapprox N_{\mathrm{res}}$. Also, the lower bound is assumed to be sufficiently large, $N_\mathrm{min} > N_\mathrm{max} - N_\mathrm{res}$. 

The resource distribution function at time slot $t_1 > t_0$ is formulated as
\begin{equation}
P_{\mathrm{res}}(N_{t_1}, t_1) = \sum_{N_{t_1-1}=0}^\infty P_\mathrm{res}(N_{t_1-1}, t_1-1) P_{\mathrm{res}}(N_{t_1} | N_{t_1-1})\,.
\end{equation}
The sum can be divided into two regions, below $N_\mathrm{res}$ and above. Since $P_{\mathrm{res}}(N, t) \rightarrow 0$ for any $N > N_\mathrm{max}$ and $N_\mathrm{max}$ is close to the number of resources of the system, we approximate the conditional function by assuming $N_\mathrm{res}$ resources in the previous time slot. For a moderately loaded system, this is a valid assumption, since the resource probability distribution function is decreasing fast for $N > N_\mathrm{res}$. Only for small arguments $N_t$ close to $0$ the deviation increases. However, the constraint regarding $N_\mathrm{min}$, which prevents underutilization, ensures that $P_\mathrm{res}(N_{t_1} | N_{t_1-1})$ is getting very small in that region anyway. Hence, we approximate the conditional resource distribution probability for $N_{t-1} > N_\mathrm{res}$ by
\begin{equation}\label{eq:apprx_cres}
P_{\mathrm{res}}(N_{t} | N_{t-1}) \approx 
\sum_{n = 0}^{N_{t}} P_{\mathrm{A}}(n)
\cdot P_{\mathrm{H}}(N_{t} - n|N_{\mathrm{res}})\,.
\end{equation}
\end{proof}

Using Lemma~\ref{lem:approx_res} for $\Delta_t > 0$, the scheduling probability is approximated by
\ifone
\begin{align}
P_1(t_0, \Delta_t)  &\approx \sum_{k_0 = N_{\mathrm{res}}}^{N_\mathrm{max}} P_{\mathrm{res}}(k_0, t_0)\left(1 - \frac{N_{\mathrm{res}}}{k_0+1}\right) \left( \sum_{k_1 = N_{\mathrm{res}}}^{N_\mathrm{max}} P_{\mathrm{res}}(k_1 | N_\mathrm{res})\left(1 - \frac{N_{\mathrm{res}}}{k_1+1}\right)\right)^{(\Delta_t - 1)}\nonumber\\
&\quad\quad\left( \sum_{k=0}^{N_{\mathrm{res}}-1} P_{\mathrm{res}}(k|N_\mathrm{res}) + \sum_{k = N_{\mathrm{res}}}^{N_\mathrm{max}} P_{\mathrm{res}}(k |N_\mathrm{res})\frac{N_{\mathrm{res}}}{k+1} \right)\,.
\end{align}
\else
\begin{align}
P_1&(t_0, \Delta_t)  \approx \sum_{k_0 = N_{\mathrm{res}}}^{N_\mathrm{max}} P_{\mathrm{res}}(k_0, t_0)\left(1 - \frac{N_{\mathrm{res}}}{k_0+1}\right) \nonumber\\
&\left( \sum_{k_1 = N_{\mathrm{res}}}^{N_\mathrm{max}} P_{\mathrm{res}}(k_1 | N_\mathrm{res})\left(1 - \frac{N_{\mathrm{res}}}{k_1+1}\right)\right)^{(\Delta_t - 1)}\nonumber\\
&\left( \sum_{k=0}^{N_{\mathrm{res}}-1} P_{\mathrm{res}}(k|N_\mathrm{res}) + \sum_{k = N_{\mathrm{res}}}^{N_\mathrm{max}} P_{\mathrm{res}}(k |N_\mathrm{res})\frac{N_{\mathrm{res}}}{k+1} \right)\,.
\end{align}
\fi

\section{Resource distribution function of a system with finite resources}
\label{sec:res_distr}
The resource distribution function describes the probability of having a specific number of resources $N$ to be scheduled at a time slot $t$. With the \replaced{previously mentioned}{aforementioned} system setup mainly three components contribute to resource allocations. The first are the packet arrival processes of the individual \acp{UE}. These pose the main component. Additionally, there are the \ac{HARQ} retransmissions, which depend on the error probability of the underlying channel code for a specific channel. However, to simplify analysis a uniform \ac{BLER} has been assumed for each of the transmissions. The last component is the overload of the previous time slot due to resource shortage, which is then transfered to the next time slot. Hence, the resource distribution is described as 
\begin{equation}\label{eq:resd}
P_{\mathrm{res}}(N, t) = \sum_{n,m,o \in \mathcal{S}} P_{\mathrm{A}}(n) P_{\mathrm{H}}(m, t-T_{\mathrm{RTT}}) P_{\mathrm{OL}}(o, t - 1)\,,
\end{equation}
with $\mathcal{S} = \{n,m,o \in \mathbb{N}_0 : n+m+o = N\}$, $N \in \mathbb{N}_0$ and $t \in \mathbb{Z}$ and $P_{\mathrm{A}}(n)$ being the probability of having $n$ arrival processes, $P_{\mathrm{H}}(m)$ being the probability of having $m$ \ac{HARQ} retransmissions in time slot $t$ and $P_{\mathrm{OL}}(o, t)$ being the probability of having $o$ resources overload in the time slot $t$ to be transferred to the next time slot.

The probability of arrival processes for $N_{\mathrm{UE}}$ \acp{UE} is described straightforwardly as a binomial distribution for $n \leq N_{\mathrm{UE}}$ i.e.\ 
\begin{equation}
P_{\mathrm{A}}(n) = 
\binom{N_{\mathrm{UE}}}{n} (P_\mathrm{A,UE})^{n} (1 - P_\mathrm{A,UE})^{N_{\mathrm{UE}}-n}\,,
\end{equation}
and otherwise $P_{\mathrm{A}}(n) = 0$,
where $P_\mathrm{A,UE}$ is the probability of packet arrival of one \ac{UE} at one time slot. This modeling implicitly assumes that one \ac{UE} can only have at most one new transmission per time slot.

Formulating $P_{\mathrm{H}}$ is a bit more intricate since for a limited allowed number of \ac{HARQ} retransmissions initial packet transmissions have to be distinguished probability-wise from \ac{HARQ} retransmissions. This would require to distinguish initial transmissions and first, second up to $n$ retransmissions as separate dependencies in $P_\mathrm{res}$ and would require to specify scheduling rules, which would considerably complicate the whole analysis. However, this limitation is overcome by allowing unlimited \ac{HARQ} retransmissions. This implies that this approach cannot be used to analyze for example single-retransmission HARQ since the HARQ retransmission term assuming an infinite number of retransmissions as implemented below would drastically overestimate the system load from HARQ retransmissions hence punishing FPR too much. Hence, $P_{\mathrm{H}}$ is given for $t \geq 0$ and $n \leq N_{\mathrm{res}}$ as
\ifone
\begin{align}
P_{\mathrm{H}}(n, t) =& 
\sum_{k=n}^{\infty} P_\mathrm{res}(k,t) \binom{\min(k,N_{\mathrm{res}})}{n}\cdot {P_\mathrm{r}}^n (1 - P_\mathrm{r})^{(\min(k,N_{\mathrm{res}})-n)}\,,
\end{align}
\else
\begin{align}
P_{\mathrm{H}}(n, t) =& 
\sum_{k=n}^{\infty} P_\mathrm{res}(k,t) \binom{\min(k,N_{\mathrm{res}})}{n}\nonumber\\
&\cdot {P_\mathrm{r}}^n (1 - P_\mathrm{r})^{(\min(k,N_{\mathrm{res}})-n)}\,,
\end{align}
\fi
and otherwise $P_{\mathrm{H}}(n, t) = 0$ except for $P_{\mathrm{H}}(0, t<0) = 1$,
where $N_{\mathrm{res}}$ is the number of system resources per time slot, $T_{RTT}$ is the \ac{HARQ} \ac{RTT} and the single-retransmission probability $P_\mathrm{r} = (1 - P_\mathrm{fn})P_\mathrm{e} + P_\mathrm{fp}(1-P_\mathrm{e})$ as in (\ref{eq:ETr}). Because of notational reasons, we chose using an infinite sum, which can be easily replaced by splitting the sum at $N_{\mathrm{res}}$ and replacing the part from $N_{\mathrm{res}}+1$ to $\infty$ by $\left( 1 - \sum_{k=0}^{N_{\mathrm{res}}} P_\mathrm{res}(k,t-T_{RTT}) \right)\binom{N_{\mathrm{res}}}{n} (P_\mathrm{r})^n (1 - P_\mathrm{r})^{(N_{\mathrm{res}}-n)}$. Still, this way of evaluating the HARQ-contributions in the system still overestimates the load from retransmissions and therefore underestimates the system performance.

The last component $P_{\mathrm{OL}}$ is simply defined by a back reference to the resource distribution function in the previous slot i.e.\ 
\begin{equation}
P_{\mathrm{OL}}(n, t) = 
\begin{cases}
\ifone
\vspace{-5pt}
\fi
P_{\mathrm{res}}(N_{\mathrm{res}} + n, t)\,, & \text{if } t \geq 0 \wedge n > 0\\\ifone\vspace{-5pt}\fi
\sum_{k=0}^{N_{\mathrm{res}}} P_{\mathrm{res}}(n, t)\,, & \text{if } t \geq 0 \wedge n = 0\\\ifone\vspace{-5pt}\fi
1\,, & \text{if } t < 0 \wedge n = 0\\\ifone\vspace{-5pt}\fi
0\,, & \text{otherwise}
\end{cases}\,.
\end{equation}

For the sake of simplicity, we may assume $T_{\mathrm{RTT}} = 1$. This assumption makes the resource distribution function at time slot $t$ only dependent on the previous time slot $t-1$ and is a valid assumption for the evaluated early \ac{HARQ} schemes. 

\begin{figure}[!t]
\centering     
\ifone
\subfigure[$N_{\mathrm{UE}} = 30$]{\label{fig:resdistr_ovrld}\includegraphics[width=.28\columnwidth]{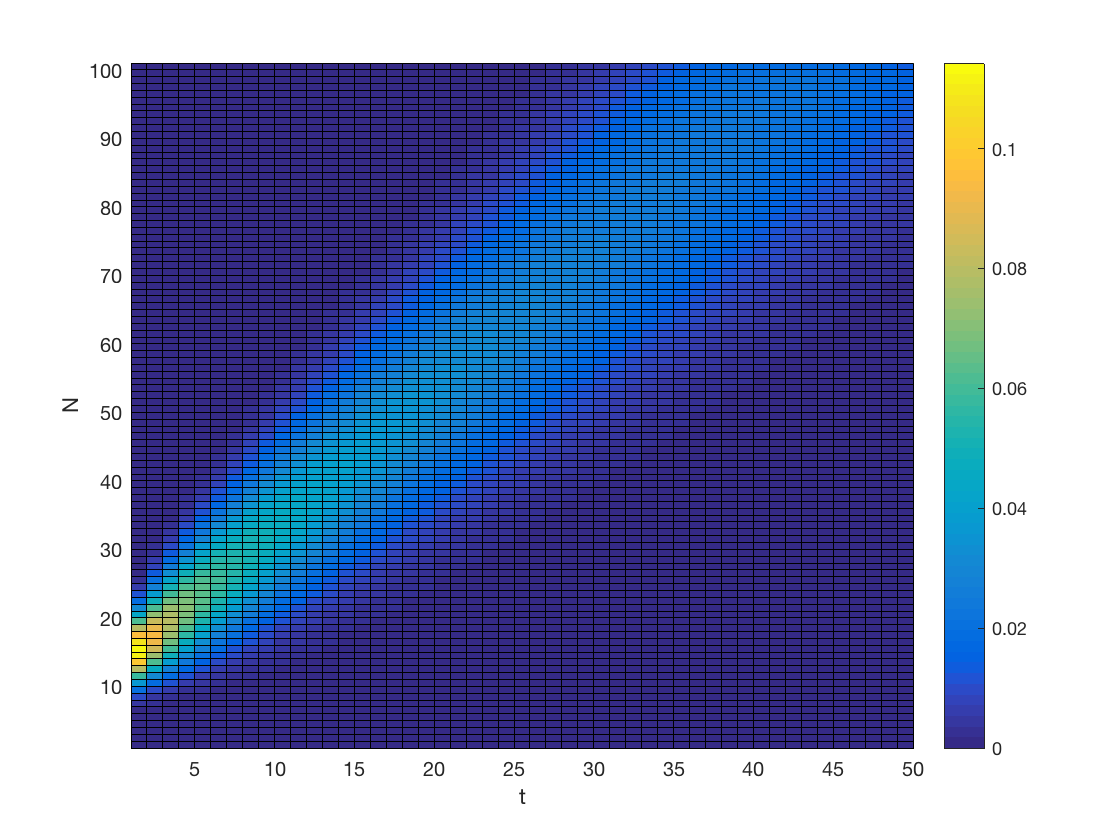}}
\subfigure[$N_{\mathrm{UE}} = 20$]{\label{fig:resdistr_balanced}\includegraphics[width=.28\columnwidth]{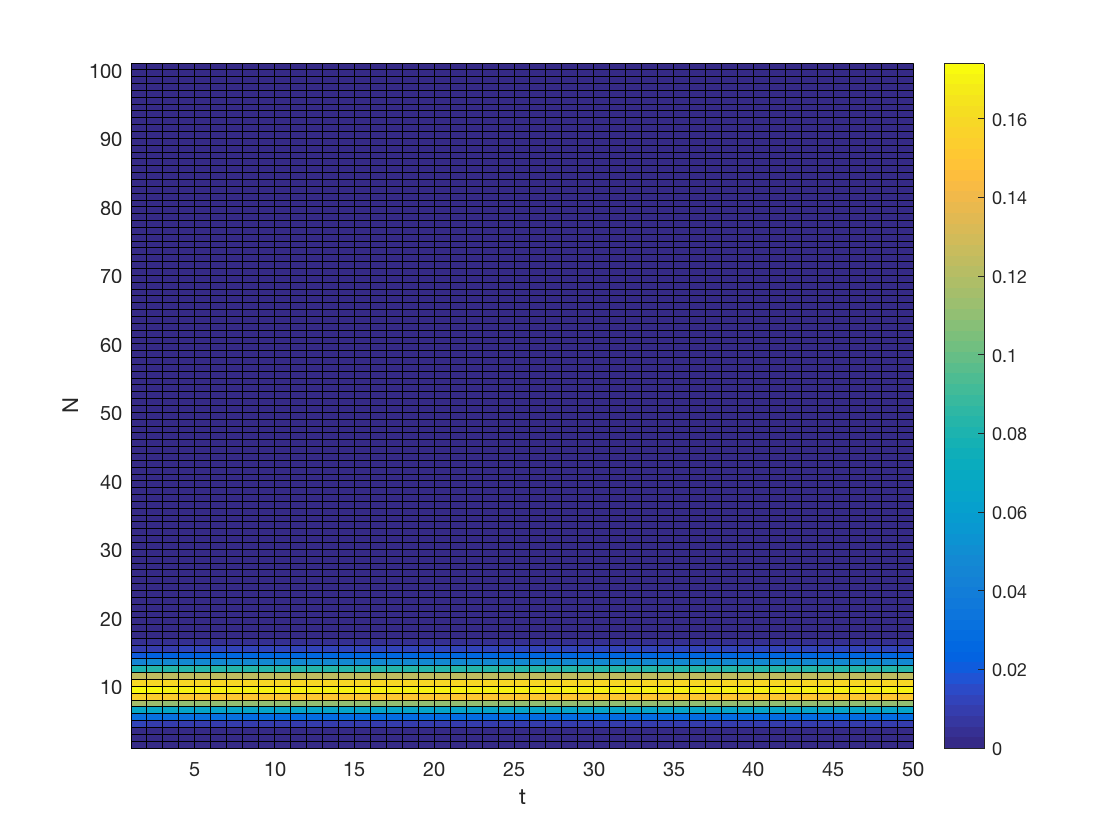}}
\else
\subfigure[$N_{\mathrm{UE}} = 30$]{\label{fig:resdistr_ovrld}\includegraphics[width=.48\columnwidth]{resource_distribution_overload.png}}
\subfigure[$N_{\mathrm{UE}} = 20$]{\label{fig:resdistr_balanced}\includegraphics[width=.48\columnwidth]{resource_distribution.png}}
\fi
\caption{Non-converging and converging resource distribution functions over time of an overloaded system (left) and a balanced  system (right).}
\label{fig:resdistr}
\end{figure}

Here, the interesting question is, if the resource distribution converges for $t \rightarrow \infty$. By simulating the propagation of $P_{\mathrm{res}}(N, t)$ over $t$, we gain an insight on that question, as presented in Figure~\ref{fig:resdistr}. As obvious in Figure~\ref{fig:resdistr_ovrld}, choosing the parameters such that the system is massively overloaded results in divergence of the resource distribution function. However, in case of a balanced system the resource distribution function shows a strong convergence behavior, as noticeable in Figure~\ref{fig:resdistr_balanced}.
From (\ref{eq:resd}), the conditioned resource distribution function for $t>0$ and $N_{t} \geq N_{t-1}-N_{\mathrm{res}}$ follows as \enlargethispage{2ex}
\begin{equation}\label{eq:cond_resd}
P_{\mathrm{res}}(N_{t} | N_{t-1}) = 
\sum_{n = 0}^{N_{\mathrm{up}}} P_{\mathrm{A}}(n)\\
P_{\mathrm{H}}(N_{\mathrm{up}} - n|N_{t-1})\,,
\end{equation}
where $N_{\mathrm{up}} = N_{t} - \max(N_{t-1}-N_{\mathrm{res}},0)$ and for $m \leq \min(N_{\mathrm{res}},N_{t-1})$
\ifone
\begin{align}
P_{\mathrm{H}}(m | N_{t-1}) =& 
\binom{\min(N_{t-1},N_{\mathrm{res}})}{m}(P_\mathrm{r})^m (1 - P_\mathrm{r})^{(\min(N_{t-1},N_{\mathrm{res}})-m)}\,,
\end{align}
\else
\begin{align}
P_{\mathrm{H}}(m | N_{t-1}) =& 
\binom{\min(N_{t-1},N_{\mathrm{res}})}{m}\nonumber\\
&\cdot (P_\mathrm{r})^m (1 - P_\mathrm{r})^{(\min(N_{t-1},N_{\mathrm{res}})-m)}\,,
\end{align}
\fi
otherwise $P_{\mathrm{H}}(m | N_{t-1}) = 0$.

\end{document}